\documentclass[journal,twoside,web]{ieeecolor}

\usepackage{generic}
\usepackage{cite}
\usepackage{amsmath,amssymb,amsfonts}
\usepackage{algorithmic}
\usepackage{graphicx}
\usepackage{textcomp}
\def\BibTeX{{\rm B\kern-.05em{\sc i\kern-.025em b}\kern-.08em
    T\kern-.1667em\lower.7ex\hbox{E}\kern-.125emX}}

\usepackage{mathtools}
\usepackage{nccmath}
\usepackage{makecell}
\let\labelindent=\undefined
\usepackage{enumitem}
\usepackage{subfigure}

\newtheorem{thm}{Theorem}
\newtheorem{lem}[thm]{Lemma}
\newtheorem{defn}[thm]{Definition}
\newtheorem{cor}[thm]{Corollary}
\newtheorem{exmp}[thm]{Example}

\makeatletter
\newenvironment{subtheorem}[1]{%
  \def\subtheoremcounter{#1}%
  \refstepcounter{#1}%
  \protected@edef\theparentnumber{\csname the#1\endcsname}%
  \setcounter{parentnumber}{\value{#1}}%
  \setcounter{#1}{0}%
  \expandafter\def\csname the#1\endcsname{\theparentnumber\alph{#1}}%
  \ignorespaces
}{%
  \setcounter{\subtheoremcounter}{\value{parentnumber}}%
  \ignorespacesafterend
}
\makeatother
\newcounter{parentnumber}

\newtheorem{assumption}{Assumption}

\DeclarePairedDelimiter{\norm}{\lVert}{\rVert}
\DeclareMathOperator{\Null}{\text{Null}}
\DeclareMathOperator{\rank}{\text{rank}}

\DeclareMathOperator{\Dim}{\text{dim}}
\DeclarePairedDelimiter{\Span}{\text{span}\{}{\}}
\DeclarePairedDelimiter{\diag}{\text{diag}\{}{\}}
\DeclareMathOperator{\allzero}{\mathbf{0}}
\DeclareMathOperator{\allone}{\mathbf{1}}

\newcommand{\ECF}[3]{\ensuremath{f_{#1}(s_{#1}(#2),#3)}}
\newcommand{\CF}[2]{\ensuremath{F_c(S(#1),#2)}}
\newcommand{\EF}[2]{\ensuremath{F_1(S'(#1),#2)}}

\begin{document}
\title{Clock Rigidity and Joint Position-Clock Estimation in Ultra-Wideband Sensor Networks}
\author{Ruixin Wen, Eric Schoof, and Airlie Chapman
\thanks{This work is supported by the University of Melbourne through Melbourne Research Scholarship.}
\thanks{R. Wen is with the Department of Mechanical Engineering, The University of Melbourne, Victoria 3010, Australia (e-mail: ruixinw1@student.unimelb.edu.au). }
\thanks{A. Chapman is with the Department of Mechanical Engineering, The University of Melbourne, Victoria 3010, Australia (e-mail: airlie.chapman@unimelb.edu.au).}
\thanks{E. Schoof is with the Department of Electrical and Electronic Engineering, The University of Melbourne, Victoria 3010, Australia (e-mail: eschoof@unimelb.edu.au).}}

\maketitle

\begin{abstract}
Joint position and clock estimation is crucial in many wireless sensor network applications, especially in distance-based estimation with time-of-arrival (TOA) measurement. In this work, we consider a TOA-based ultra-wideband (UWB) sensor network, propose a novel clock rigidity theory and investigate the relation between the network graph properties and the feasibility of clock estimation with TOA timestamp measurements. It is shown that a clock framework can be uniquely determined up to a translation of clock offset and a scaling of all clock parameters if and only if it is infinitesimally clock rigid. We further prove that a clock framework is infinitesimally clock rigid if its underlying graph is generically bearing rigid in $\mathbb{R}^2$ with at least one redundant edge. Combined with distance rigidity, clock rigidity provides a graphical approach for analyzing the joint position and clock problem. It is shown that a position-clock framework can be uniquely determined up to some trivial variations corresponding to both position and clock if and only if it is infinitesimally joint rigid. Simulation results are presented to demonstrate the clock estimation and joint position-clock estimation.
\end{abstract}


\section{Introduction}
\label{sec:introduction}

\IEEEPARstart{P}{osition} estimation of multi-agent systems, also referred to as network localization, is a critical aspect for a variety of robotic applications such as vehicle tracking and industrial process monitoring \cite{akyildiz2002survey}. The Global Positioning System (GPS) is widely used for robot position estimation, but it lacks precision and can fail entirely in a GNSS-constrained environment, such as inside buildings and underground locations. With the development of low-cost, low-power and multi-functional sensors, many research works focus on position estimation in wireless sensor networks, in which specialized sensors are mounted on robots and positions are estimated by using knowledge of the absolute positions of a limited number of sensors and inter-sensor measurements such as distance and bearing \cite{mao2007wireless}. 

A related problem to multi-agent position estimation is whether a given static sensor network's position information can be uniquely determined up to some trivial motions, e.g., translation, rotation and scaling. Graph theory, and in particular graph rigidity, is a useful tool for analyzing and solving this problem. The application of graph rigidity to position estimation is investigated and demonstrated in recent literature \cite{eren2004rigidity,aspnes2006theory,anderson2010formal,shames2009minimization,zelazo2012rigidity,schoof2014bearing,zelazo2015decentralized,krick2009stabilisation}. Infinitesimal rigidity is a crucial concept in rigidity theory. It defines the sufficient graph properties of a framework so that all the infinitesimal motions preserving distance (bearing) measurements are trivial.

Clock estimation is another critical aspect in wireless sensor networks, especially for the commonly-used time-of-arrival (TOA) ranging techniques in distance-based position estimation. Various algorithms for estimating network clock parameters such as clock offset and clock skew are studied in \cite{rhee2009clock,ganeriwal2003timing,noh2008new,schenato2007distributed}. Most works of network clock estimation rely on spanning trees \cite{ganeriwal2003timing} or a cluster-based structure \cite{noh2008new} of the network, but little research explores the problem from a more general graph topology perspective. 

The close relation between position and clock estimation problem necessitates a joint estimation approach. Several research works \cite{chepuri2012joint,rajan2011joint,yu2009toa,ahmad2013joint,alanwar2017d} have studied the simultaneous estimation of network position and clock information, which are mostly tied from a statistical signal processing perspective and use redundant communication and known position or clock information of some sensors (called anchors) to acquire a unique estimation result. The relationship between the graph properties of an anchor-free network and the network's capability to estimate its position and clock by one round timestamp measurements is still an open problem. 

Ultra-wideband (UWB) is a short range, high bandwidth radio technology. It uses a broad range of frequencies to generate energy pulses with sharp rising edges, which allow for highly precise signal sending and receiving timestamp measurements \cite{oppermann2005uwb}. Ultra-wideband sensors are widely used for distributed sensing in position estimation due to their high accuracy, low price and low computation complexity. 

In this paper, we consider TOA-based UWB sensor networks in which both sending and receiving timestamps are accurately measured. Analogous to the distance (bearing) measurements being preserved in distance (bearing) rigidity theory \cite{anderson2008rigid,zhao2015bearing}, we show that the timestamp measurement can be studied in a similar way for the clock estimation problem.

We propose a novel TOA-based clock rigidity theory under a bidirectional communication assumption, showing that a clock framework can be uniquely determined up to some trivial variations (a shift on clock offset and a skew on all clock parameters). We also explore the connection between clock rigidity and bearing rigidity theory, showing that a clock framework is infinitesimally clock rigid if and only if its underlying graph is generically bearing rigid in $\mathbb{R}^2$ with at least one redundant edge.

Based on the proposed clock rigidity theory, a graph-theoretic approach for studying joint position and clock estimation problem is investigated. It is proved that a position-clock framework with certain graph properties can be uniquely determined up to trivial variations corresponding to both position and clock, i.e., up to a translation and a rotation of position, a shift of clock offset and a scaling of both position and clock. 

The structure of the work is as follows. Section \ref{sec_clockrigidity} presents the TOA-based clock rigidity theory under the bidirectional communication assumption. Section \ref{sec_connection} establishes the connection between clock rigidity and bearing rigidity theory. Section \ref{sec_jointmatrix} analyzes the joint position and clock problem based on the combination of clock rigidity theory and distance rigidity theory. Section \ref{sec_estimation} applies this new theory to the clock estimation and the position-clock estimation problems using a gradient-descent method.

\emph{Notation}: Matrices are denoted by capital letters (e.g., $A$). The rank and null space of a matrix $A$ are denoted by rank$(A)$ and Null$(A)$, respectively. A diagonal matrix with diagonal entries $d_1,...,d_n$ is denoted as $\diag{d_i}$. A matrix or a vector that consists of all zero entries is denoted by $\allzero$. The vector $\allone_n$ denotes the $n\times1$ vector of all ones. The identity matrix in $\mathbb{R}^{n\times n}$ is denoted by $I_n$. The Kronecker product of two matrices (vectors) $A$ and $B$ is written as $A\otimes B$, and $\norm{\cdot}$ denotes the Euclidean norm of a vector.  An elemental rotation in $d$-dimensional space is a rotation about the $(d-2)$-dimensional subspace containing a set of $(d-2)$ vectors in the standard basis. Matrix $J_d^i$ denotes the infinitesimal generator of the $i$th elemental rotation in $d$-dimensional space, where $i\in\{1,2,...,d(d-1)/2\}$. For example, for $d=2$ and $d=3$,
\begin{align*}
    &J^1_2=\begin{bmatrix}
    0 & 1\\
    -1 & 0
    \end{bmatrix}, \\
    &J^1_3=\begin{bmatrix}
    0 & 0 & 0\\
    0 & 0 & 1\\
    0 & -1 & 0
    \end{bmatrix},
    J^2_3=\begin{bmatrix}
    0 & 0 & 1\\
    0 & 0 & 0\\
    -1 & 0 & 0
    \end{bmatrix},
    J^3_3=\begin{bmatrix}
    0 & -1 & 0\\
    1 & 0 & 0\\
    0 & 0 & 0
    \end{bmatrix}.
\end{align*}

An undirected graph,  denoted as $\mathcal{G}=(\mathcal{V},\mathcal{E})$, consists of a vertex set $\mathcal{V}=\{1,...,n\}$ and an edge set $\mathcal{E}\subseteq\mathcal{V}\times\mathcal{V}$ with cardinality $m$. Two vertices $v_i$ and $v_j$ are called neighbors when $\{v_i,v_j\}\in\mathcal{E}$. A directed graph is denoted as $\mathcal{D}=(\mathcal{V},\mathcal{E}_\mathcal{D})$, where $\mathcal{E}_\mathcal{D}$ is an directed edge set.

\section{Clock rigidity }\label{sec_clockrigidity}
In this section, we propose a clock rigidity theory under the assumption of bidirectional communication. The basic problem that this clock rigidity theory studies is whether a clock framework can be uniquely determined up to some trivial variations given the TOA timestamp measurements between each pair of neighbors in the framework. This problem can be equivalently stated as whether it can be determined that two clock frameworks with the same inter-neighbor timestamp measurements will have the same parameters for the assumed clock model.

We begin by defining the first-order affine clock model. Consider a network with $n$ nodes, in which every node has its independent clock and exhibits a constant clock offset and clock skew. Let $t_i$ be the local time measured at the $i$th node and $t$ be the global reference time. We assume that the relation between the local time and the global reference time can be given by a first-order clock model \cite{rajan2011joint},
\begin{equation}
    t_i=w_it+\phi_i \quad\Leftrightarrow \quad t=\alpha_it_i+\beta_i,
\end{equation}
where the global clock skew $w_i\in \mathbb{R}^+$ and the global clock offset $\phi_i\in\mathbb{R}$ characterize the mapping from global reference time to the local time of node $i$. The local clock skew $\alpha_i=w_i^{-1}\in\mathbb{R}^+$ and the local clock offset $\beta_i = -w_i^{-1}\phi_i\in\mathbb{R}$ characterize the mapping from local time of node $i$ to the global reference time. In practice, global clock skew is stochastic due to the clock drift, but considering most clocks drift slowly, we can assume that the local clock parameters $\alpha_i$ and $\beta_i$ are constant over short periods of time. For brevity, we refer to $\alpha_i$ and $\beta_i$ as simply the clock skew and clock offset, respectively.

Ultra-wideband technology uses a wide bandwidth to generate signals with sharp edges,  which provide highly accurate sending and receiving timestamp measurements. Consider a TOA-based UWB sensor network, a ranging signal is transmitted from one node to another and the transmission and reception timestamps are recorded independently in local time coordinates, as shown in Fig. \ref{fig_TOA}. With the assumed clock model, the inter-agent distance can be expressed as
\begin{align}\label{OWRdist}
    d_{ij}=c(\alpha_jT^j_{(i,j)}+\beta_j-\alpha_iT^i_{(i,j)}-\beta_i),
\end{align}
where $d_{ij}$ is the distance between the $i$th and $j$th node, $c$ is the speed of light, for a ranging signal sending from node $i$ to node $j$, the sending timestamp at the local time coordinate of node $i$ is denoted as $T^i_{(i,j)}$, and the receiving timestamp at the local time coordinate of node $j$ is denoted as $T^j_{(i,j)}$. 
\begin{figure}[h]
    \centering
    \includegraphics[scale=0.25]{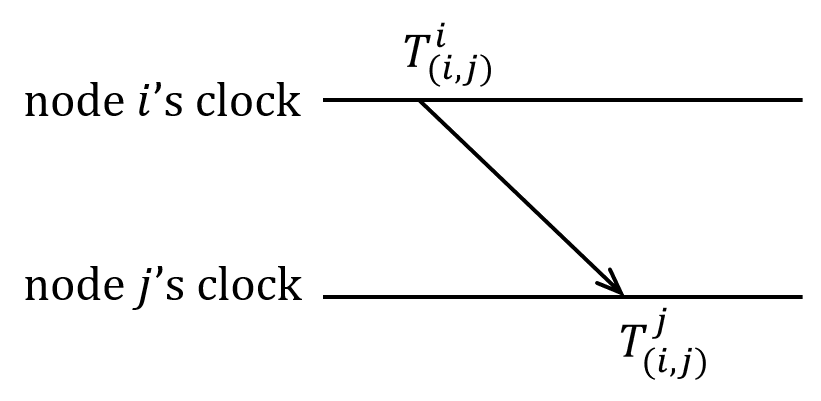}
    \caption{TOA measurement and local timestamp notation.}
    \label{fig_TOA}
\end{figure}

The communication behind (\ref{OWRdist}) is directed. In order to consider the network problem in an undirected way, we have the following assumption. 
\begin{assumption}[Bidirectional communication]\label{assump1}
  The inter-node communication is bidirectional, i.e., the $i$th node can receive a ranging signal from the $j$th node if and only if the $j$th node can receive a ranging signal from the $i$th node.
\end{assumption}

Assumption \ref{assump1} assumes the symmetric visibility between neighbors, which is trivially satisfied when all nodes share a common communication range. It simplifies the network problem from a directed graph to an undirected graph, and also provides a useful equivalent relation between the neighbor nodes, which will be shown later.

Now we define some necessary notation. Given a UWB sensor network of $n$ nodes, under Assumption \ref{assump1}, it can be represented by an undirected graph $\mathcal{G}=(\mathcal{V}, \mathcal{E})$ where each vertex $v_i$ in the vertex set $\mathcal{V}$ is associated with $i$th sensor node and each edge $\{v_i,v_j\}$ in the edge set $\mathcal{E}$ corresponds to a sensor node pair which has bidirectional communication. A clock configuration is denoted as $\varphi=[\varphi_1^T,...,\varphi_n^T]^T\in\mathbb{R}^{2n}$, where $\varphi_i=[\alpha_i,\beta_i]^T$. We also define a clock skew configuration $\alpha=[\alpha_1^T,...,\alpha_n^T]^T$ and a clock offset configuration $\beta=[\beta_1^T,...,\beta_n^T]^T$. A clock framework, denoted as $(\mathcal{G},\varphi)$, is a combination of an undirected graph $\mathcal{G}$ and a clock configuration $\varphi$, which provides a mapping from vertex $v_i\in\mathcal{V}$ to the parameter $\varphi_i$. Note that for a static UWB sensor network, we assume that $d_{ij}$ is fixed and the sending timestamp $T_{(i,j)}^i$ and $T_{(j,i)}^j$ are known for all $\{v_i,v_j\}\in\mathcal{E}$. So the receiving timestamp measurement $T_{(i,j)}^j$ and $T_{(j,i)}^i$ can be uniquely determined by $\varphi_i$ and $\varphi_j$.

Under Assumption \ref{assump1} and the distance relation (\ref{OWRdist}), TOA ranging measurements from both ends of one edge are available and equal, i.e., $d_{ij}=d_{ji}$. Rewriting this distance equivalence with the timestamp notation defined above, we have
\begin{align}\label{ts_equal}
\begin{split}
    \alpha_jT^j_{(i,j)}+\beta_j-\alpha_iT^i_{(i,j)}-\beta_i=\alpha_iT^i_{(j,i)}+\beta_i-\alpha_jT^j_{(j,i)}-\beta_j,
    \end{split}
\end{align}
which can be written in the following form for every $\{v_i,v_j\}\in \mathcal{E}$:
\begin{align} \label{clock_constr}
    \alpha_i\overline{T}^i_{ij}+\beta_i-\alpha_j\overline{T}^j_{ij}-\beta_j=0, 
\end{align}
where  
\begin{align}
     \overline{T}^i_{ij}=\frac{T^i_{(i,j)}+T^i_{(j,i)}}{2}\quad\text{and} \quad\overline{T}^j_{ij}=\frac{T^j_{(i,j)}+T^j_{(j,i)}}{2}.
 \end{align}
Define the edge clock function $f_{ij}:\mathbb{R}^{2n}\rightarrow\mathbb{R}$ as
\begin{align} \label{edge_clock_func}
   \ECF{ij}{\varphi}{\varphi}=s_{ij}(\varphi)^T\varphi,
\end{align}
where
\begin{align}\label{sij}
    s_{ij}(\varphi)=  \Bigg[\allzero^T\quad\underbrace{ \overline{T}^i_{ij}\quad1}_{v_i}\quad\allzero^T\quad\underbrace{-\overline{T}^j_{ij}\quad-1}_{v_j}\quad \allzero^T\Bigg]^T.
\end{align}
Equation (\ref{clock_constr}) can be equivalently written as
\begin{equation}\label{fk_constr}
\ECF{ij}{\varphi}{\varphi}=0.
\end{equation}

We are now ready to define the fundamental concepts in clock rigidity. These concepts are defined analogously to those in the distance rigidity theory \cite{anderson2008rigid} and bearing rigidity theory \cite{zhao2015bearing}.

\begin{defn} \label{def_clockequivalency}
Clock frameworks $(\mathcal{G},\varphi)$ and $(\mathcal{G},\varphi')$ are \textbf{clock equivalent} if $\ECF{ij}{\varphi}{\varphi'}=0$ for all $\{v_i,v_j\}\in\mathcal{E}$.
\end{defn}

\begin{defn}\label{def_clockcongruency}
Clock frameworks $(\mathcal{G},\varphi)$ and $(\mathcal{G},\varphi')$ are \textbf{clock congruent} if $\ECF{ij}{\varphi}{\varphi'}=0$ for all $v_i,v_j\in\mathcal{V}$.
\end{defn}
\begin{defn}\label{def_clockrigidity}
A clock framework $(\mathcal{G},\varphi)$ is \textbf{clock rigid} if there exists a constant $\epsilon>0$ such that any clock framework $(\mathcal{G},\varphi')$ that is clock equivalent to $(\mathcal{G},\varphi)$ and satisfies $\norm{\varphi'-\varphi}<\epsilon$ is also clock congruent to $(\mathcal{G},\varphi)$ .
\end{defn}
\begin{defn}\label{def_globalclockrigidity}
A clock framework $(\mathcal{G},\varphi)$ is \textbf{globally clock rigid} if an arbitrary clock framework that is clock equivalent to $(\mathcal{G},\varphi)$ is also clock congruent to $(\mathcal{G},\varphi)$.
\end{defn}

We next define infinitesimal clock rigidity, which is one of the most important concepts in the clock rigidity theory. 

For convenience, we also reference the edge clock function $f_{ij}$ and $s_{ij}$ using an edge index $k$ rather than node pair $\{v_i,v_j\}\in\mathcal{E}$ for some edge ordering $k\in\{1,...,m\}$ as
\begin{align}
    f_k\triangleq f_{ij}, s_k\triangleq s_{ij},
\end{align}
and define the matrix
\begin{align}
    S(\varphi)=[s_1(\varphi),...,s_m(\varphi)]^T.
\end{align}
Now we define the clock function $F_c:\mathbb{R}^{2n}\rightarrow\mathbb{R}^m$ as
\begin{align}\label{clockfunc}
   \CF{\varphi}{\varphi}\triangleq[\ECF{1}{\varphi}{\varphi},...,\ECF{m}{\varphi}{\varphi}]^T,
\end{align}
so the constraint (\ref{clock_constr}) can be written as
\begin{align}\label{clockfunc_constr}
   \CF{\varphi}{\varphi}=S(\varphi)\varphi=\allzero.
\end{align}

The clock function describes all the clock constraints in the clock  framework. We define the clock rigidity matrix as the Jacobian of the clock function,
\begin{align}\label{clockrigiditymatrix}
   R_c(\varphi)\triangleq\frac{\partial \CF{\varphi}{\varphi}}{\partial \varphi}\in\mathbb{R}^{m\times 2n}.
\end{align}
Let $\delta \varphi$ be a variation of the configuration $\varphi$. If $R_c(\varphi)\delta \varphi = 0$, then $\delta \varphi$ is called an infinitesimal clock variation of $(\mathcal{G},\varphi)$. This is analogous to infinitesimal motions in distance rigidity \cite{anderson2008rigid} and bearing rigidity \cite{zhao2015bearing}. Distance preserving motions of a framework include translations and rotations. Bearing preserving motions of a framework include translations and scalings. For a clock framework, timestamp preserving variations include translations (a common shift) on the clock offset configuration $\beta$ and scalings (a common skew) of the entire clock framework. An infinitesimal clock variation is called trivial if it corresponds to a translation of the clock offset configuration $\beta$ and a scaling of the entire clock framework. See Fig. \ref{fig_clock_trivial}.

\begin{figure}[!t]
\begin{center}
\subfigure[Translation of $\beta$]{\includegraphics[scale=0.3]{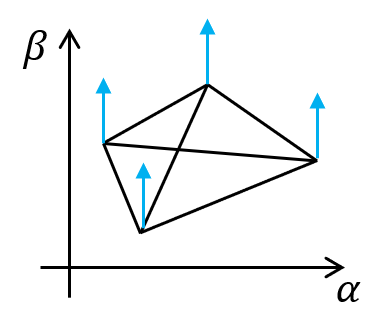}%
\label{Fig_clock_trivial_a}}
\hfil
\subfigure[Scaling of $\varphi$]{\includegraphics[scale=0.3]{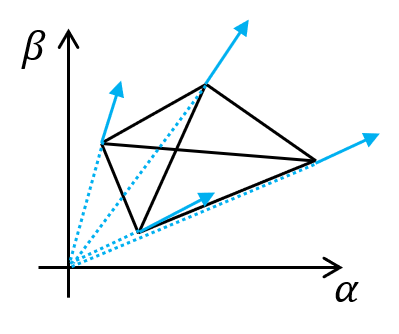}%
\label{Fig_clock_trivial_b}}
\end{center}
\caption{Basic trivial infinitesimal clock variations.}
 \label{fig_clock_trivial}
\end{figure}

\begin{defn}\label{def_infinitesimalclockrigidity}
A clock framework is \textbf{infinitesimally clock rigid} if all the infinitesimal clock variations are trivial.
\end{defn}

Up to this point, we have introduced the fundamental concepts in clock rigidity theory. We next connect these concepts using the clock rigidity matrix.

\begin{lem}\label{lem_S}
For a clock framework $(\mathcal{G},\varphi)$, the clock rigidity matrix in (\ref{clockrigiditymatrix}) can be expressed as $R_c(\varphi)=S(\varphi)$ and satisfies $R_c(\varphi)\varphi=\allzero$.
\end{lem}
\begin{proof}
It follows from the definition of the clock function and clock rigidity matrix in (\ref{clockfunc}) and (\ref{clockrigiditymatrix}). The $k$th row of clock rigidity matrix can be written as
\begin{align}
    \frac{\partial \ECF{k}{\varphi}{\varphi}}{\partial \varphi}=s_k(\varphi), 
\end{align}
for all $k\in\{1,...,m\}$. So, $ R_c(\varphi)=[s_1(\varphi),...,s_m(\varphi)]^T=S(\varphi)$. Following from the constraint in (\ref{clockfunc_constr}), $R_c(\varphi)\varphi=S(\varphi)\varphi=\allzero$ 
\end{proof}
\begin{lem}\label{lem_rankinequality}
A clock framework $(\mathcal{G},\varphi)$ satisfies $\text{span}\{\allone_n\otimes$ $[0,1]^T,\varphi\}\subseteq\Null(R_c(\varphi))$ and $\rank(R_c(\varphi))\leq2n-2$.
\end{lem}
\begin{proof}
By Lemma \ref{lem_S}, it is clear that $\varphi\subseteq\Null(R_c(\varphi))$. The expression of $s_k$ shown in (\ref{sij}) indicates that $\allone_n\otimes[0,1]^T\subseteq\Null(R_c(\varphi))$. Considering the assumption that $\alpha_i\in\mathbb{R}^+$, for all $i\in\{1,...,n\}$, the vector $\varphi$ and $\allone_n\otimes[0,1]^T$ are linearly independent. The inequality $\rank(R_c(\varphi))\leq2n-2$ follows immediately from $\Span{\allone_n\otimes[0,1]^T,\varphi}\subseteq\Null(R_c(\varphi))$. 
\end{proof}

For any undirected graph $\mathcal{G}=(\mathcal{V},\mathcal{E})$, denote $K_n$ as the $n$-node complete graph over the same vertex set $\mathcal{V}$, and $R_c^K(\varphi)$ as the clock rigidity matrix of the clock framework $(K_n,\varphi)$. The next result gives the necessary and sufficient conditions for clock equivalency and clock congruency.
\begin{thm}\label{thm_equivalencecondition}
Two clock frameworks $(\mathcal{G},\varphi)$ and $(\mathcal{G},\varphi')$ are clock equivalent if and only if $R_c(\varphi)\varphi'=0$, and clock congruent if and only if $R_c^K(\varphi)\varphi'=0$.
\end{thm}
\begin{proof}
Following from Lemma \ref{lem_S}, we have $F_\varphi(\varphi)=0$ and 
\begin{align}
    R_c(\varphi)\varphi'=0 \Leftrightarrow F_\varphi(\varphi')=0,
\end{align}
so $f_k^\varphi(\varphi)=f_k^\varphi(\varphi')=0$, for all $k\in\{1,...,m\}$. By Definition \ref{def_clockequivalency}, the two clock frameworks are clock equivalent if and only if $R_c(\varphi)\varphi'=0$. By Definition \ref{def_clockcongruency}, it can be similarly shown that clock frameworks are clock congruent if and only if $R_c^K(\varphi)\varphi'=0$.
\end{proof}

Since any infinitesimal variation $\delta \varphi$ is in $\Null(R_c(\varphi))$, Theorem \ref{thm_equivalencecondition} implies that $R_c(\varphi)(\varphi+\delta\varphi)=0$ and hence $(\mathcal{G},\varphi+\delta\varphi)$ is clock equivalent to $(\mathcal{G},\varphi)$.

It is worth noting that Theorem \ref{thm_equivalencecondition} for clock rigidity has an analogous expression in bearing rigidity theory; see Theorem 1 in \cite{zhao2015bearing}. It indicates a possible relationship between clock rigidity and bearing rigidity, which will be further discussed later. Here we provide Theorems  \ref{thm_globalrigid}-\ref{thm_CRvariation} as a straightforward extension of Theorem \ref{thm_equivalencecondition} following the corresponding bearing rigidity proofs in \cite{zhao2015bearing} (Theorems 2-6). Due to space limitations, we refer the reader to this work for detailed proofs.

\begin{thm}\label{thm_globalrigid}
A clock framework $(\mathcal{G},\varphi)$ is globally clock rigid if and only if $\Null(R_c^K(\varphi))=\Null(R_c(\varphi))$ or equivalently $\rank(R_c^K(\varphi))=\rank(R_c(\varphi))$.
\end{thm}
\begin{thm}
A clock framework $(\mathcal{G},\varphi)$ is clock rigid if and only if it is globally clock rigid.
\end{thm} 
\begin{thm}\label{thm_ICR}
For a clock framework $(\mathcal{G},\varphi)$, the following statements are equivalent:
\begin{enumerate}[label=(\alph*)]
    \item $(\mathcal{G},\varphi)$ is infinitesimally clock rigid;
    \item $\rank(R_c(\varphi))=2n-2$;
    \item $\Null(R_c(\varphi))=\Span{\allone_n\otimes[0,1]^T,\varphi}$.
\end{enumerate}
\end{thm}
\begin{thm}\label{thm_infinitesimaltoglobal}
Infinitesimal clock rigidity implies global clock rigidity.
\end{thm}

\begin{thm}\label{thm_CRvariation}
An infinitesimally clock rigid framework can be uniquely determined up to a translation of clock offset $\beta$ and a scaling of the entire clock framework.
\end{thm}

Similar to bearing rigidity, the clock rigidity of a clock framework is also a generic property which depends on its graph rather than the clock configuration. However, bearing measurements in bearing rigidity are only determined by the position configuration, whereas timestamp measurements in clock rigidity are not uniquely determined for a given clock configuration; it also depends on the distance between sensor nodes and the sending timestamp. We define a vector pair $(d_{in},T_s)$ where $d_{in}=[d_1,...,d_m]^T$ denotes the inter-node distances and $T_s=[T_{(i,j)}^i,T_{(j,i)}^j,...]^T$ for all $\{v_i,v_j\}\in\mathcal{E}$ denotes the sending timestamps. Given $(d_{in},T_s)$, we can define a generically clock rigid graph and a generic clock configuration, then show the following relationship to infinitesimal clock rigidity. The proof follows a similar structure from \cite{zhao2017laman}. 

\begin{defn}\label{def_GCR}
Given $(d_{in},T_s)$, a graph $\mathcal{G}$ is \textbf{generically clock rigid} if there exists at least one clock configuration $\varphi$ such that $(\mathcal{G},\varphi)$ is infinitesimally clock rigid.
\end{defn}
\begin{defn}
A clock configuration $\varphi$ is \textbf{generic} for graph $\mathcal{G}$ if $(\mathcal{G},\varphi)$ is infinitesimally clock rigid.
\end{defn}
\begin{lem}
 Given $(d_{in},T_s)$, if $\mathcal{G}$ is generically clock rigid, then $(\mathcal{G},\varphi)$ is infinitesimally clock rigid for almost all $\varphi$.
\end{lem}
\begin{proof}
Let $\Omega_c$ be the set of $\varphi$ where $\rank(R_c)<2n-2$. Suppose $g(\varphi)$ is the vector consisting of all the minors of $R_c$ of order $(2n-2)$. Then $\Omega_c$ is the set of solutions to $g(\varphi)=0$. By (\ref{OWRdist}), for node $j$ in any edge $\{v_i,v_j\}\in\mathcal{E}$, $T_{ij}^j(\varphi)=(d_{ij}/c+\beta_i-\beta_j+\alpha_iT_{(i,j)}^i+\alpha_jT_{(j,i)}^j)/\alpha_j$ and $\alpha_j>0$. Given $(d_{in},T_s)$, equation $g(\varphi)=0$ can be converted to a set of polynomial equations of $\varphi$. So $\Omega_c$ is an algebraic set and hence it is the entire space or it is of measure zero. According to Definition \ref{def_GCR}, there exists at least one $\varphi$ such that $(\mathcal{G},\varphi)$ is infinitesimally clock rigid, so $\Omega_c$ is not the entire space. Therefore, $\Omega_c$ is of measure zero and $(\mathcal{G},\varphi)$ is infinitesimally clock rigid for almost all $\varphi$.
\end{proof}

\section{Connection to Bearing Rigidity}\label{sec_connection}
The clock rigidity theory studies whether a clock framework can be uniquely determined by the inter-neighbor TOA timestamp measurements, which follows from research directions in distance and bearing rigidity theory \cite{anderson2008rigid,zhao2015bearing}. In this section, we establish the connection between the clock rigidity theory and the bearing rigidity theory and prove that a clock framework is infinitesimally clock rigid if and only its graph is generically bearing rigid in $\mathbb{R}^2$ with at least one redundant edge. 

To establish this connection, first we introduce the dummy variable $\gamma=[\gamma_1,...,\gamma_m]^T$ to rewrite the constraint (\ref{clock_constr}) for every $\{v_i,v_j\}\in\mathcal{E}$ with a corresponding fixed edge index $k$ in the following form
\begin{subequations}\label{seperate_constr}
\begin{align}
    &\overline{T}^i_{ij}\alpha_i+\beta_i-\gamma_k=0\\
    &\overline{T}^j_{ij}\alpha_j+\beta_j-\gamma_k=0.
\end{align}
\label{dummy}
\end{subequations}
Equation (\ref{seperate_constr}) can be also written as
\begin{subequations}
\begin{align}
    &[\overline{T}^i_{ij},1][\alpha_i,\beta_i-\gamma_k]^T=0\\
    &[\overline{T}^j_{ij},1][\alpha_j,\beta_j-\gamma_k]^T=0.
\end{align}
\end{subequations}
The vector $[\alpha_i,\beta_i-\gamma_k]^T$ and $[\alpha_j,\beta_j-\gamma_k]^T$ are orthogonal to the vector $[\overline{T}^i_{ij},1]^T$ and $[\overline{T}^j_{ij},1]^T$, respectively. Since $\alpha_i,\alpha_j>0$, for every $\{v_i,v_j\}\in \mathcal{E}$ we have
\begin{subequations}\label{bearing_constr}
\begin{align}
   &\frac{[\alpha_i,\beta_i-\gamma_k]^T}{\norm{[\alpha_i,\beta_i-\gamma_k]}}=\frac{[1,-\overline{T}^i_{ij}]^T}{\norm{[1,-\overline{T}^i_{ij}]}}\\
   &\frac{[\alpha_j,\beta_j-\gamma_k]^T}{\norm{[\alpha_j,\beta_j-\gamma_k]}}=\frac{[1,-\overline{T}^j_{ij}]^T}{\norm{[1,-\overline{T}^j_{ij}]}}.
\end{align}
\end{subequations}

\begin{figure}
  \begin{center}
  \includegraphics[scale=0.2]{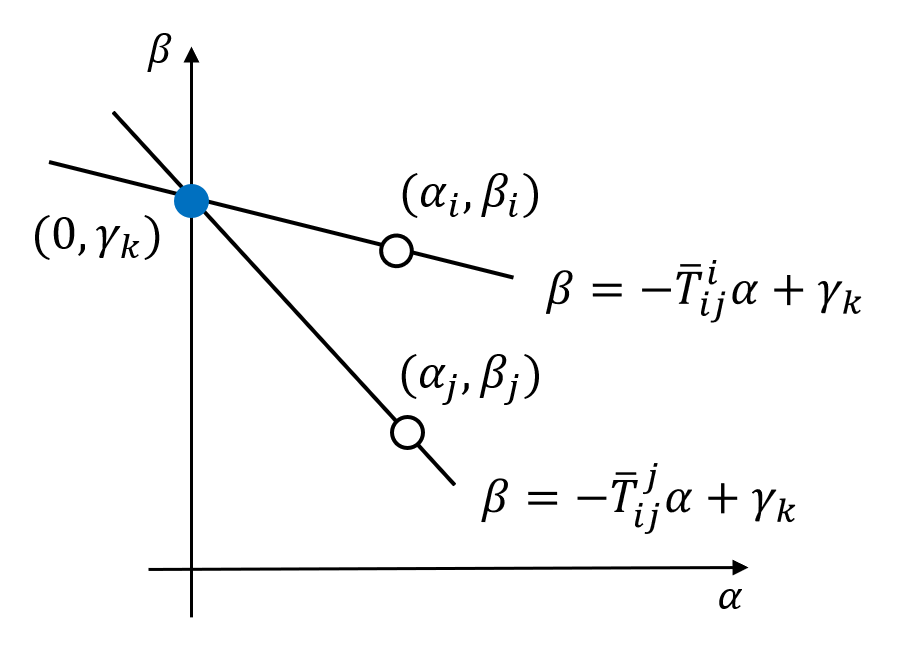}\\
  \caption{Clock constraint visualization. The constraint (\ref{seperate_constr}) can be comprehended as two straight lines with the same vertical intercept in a two-dimensional coordinate, passing through the points $(\alpha_i,\beta_i)$ and $(\alpha_j,\beta_j)$, respectively. Since the gradient of the straight line is constant, the constraint (\ref{bearing_constr}) follows.}
  \label{fig_clockconstr}
  \end{center}
\end{figure}

Fig. \ref{fig_clockconstr} visualizes the equivalence between constraint (\ref{seperate_constr}) and (\ref{bearing_constr}). We take the partial derivative of the left hand side of (\ref{bearing_constr}) with respect to $\zeta=[\varphi^T,\gamma^T]^T\in\mathbb{R}^{2n+m}$ for every $\{v_i,v_j\}\in \mathcal{E}$ and define the resulting Jacobian with respect to $\zeta$ as $S''(\zeta)\in\mathbb{R}^{4m\times (2n+m)}$. Then we can give the following lemma.

\begin{lem}\label{lem_RctoS''}
Given a clock framework $(\mathcal{G},\varphi)$ and a variation of the clock configuration $\delta\varphi$, $\delta\varphi\in\Null(R_c(\varphi))$ if and only if there exists a vector $\delta\gamma\in\mathbb{R}^m$ such that $\delta\zeta=[\delta\varphi^T,\delta\gamma^T]^T\in\Null(S''(\zeta))$. 
\end{lem}
\begin{proof}
By writing (\ref{seperate_constr}) for every $\{v_i,v_j\}\in \mathcal{E}$ as a linear matrix equation, we have
\begin{align}
   S'(\zeta)\zeta=\allzero,
\end{align}
where $S'(\zeta)\in\mathbb{R}^{2m\times (2n+m)}$.

We define extended clock function  $F_1:\mathbb{R}^{2n+m}\rightarrow\mathbb{R}^{2m}$ as
\begin{align}
   \EF{\zeta}{\zeta}\triangleq S'(\zeta)\zeta.
\end{align}
It is clear that the Jacobian with respect to $\zeta$ of the extended clock function is $S'(\zeta)$. 

Equations (\ref{seperate_constr}) and (\ref{bearing_constr}) are equivalent, so the infinitesimal variations which preserve the equality in (\ref{bearing_constr}) also preserve the equality in (\ref{seperate_constr}), i.e., $\Null(S''(\zeta))=\Null(S'(\zeta))$.

By elementary row operations which preserve the null space of matrix $S'(\zeta)$, the matrix $S'(\zeta)$ can be written in the following form:
\begin{align}
    T_2T_1S'(\zeta)=\begin{bmatrix}
        S(\varphi) & \allzero\\
        X & -I_m
    \end{bmatrix},
\end{align}
where $T_1$ is a row switching elementary matrix so that the resulting matrix $T_1S'$ has the $k$th row and $(m+k)$th row corresponding to the $k$th edge. Further, $T_2$ is a row addition elementary matrix which has the form
\begin{align*}
    T_2=\begin{bmatrix}
        I_m & -I_m\\
        \allzero & I_m
    \end{bmatrix}.
\end{align*}
Each row of the matrix $X$ corresponds to only one of the constraints (\ref{seperate_constr}) of an edge $\{v_i,v_j\}\in\mathcal{E}$, e.g., 
\begin{align}
    \Bigg[\allzero^T\quad\underbrace{0\quad0}_{v_i}\quad \allzero^T\quad\underbrace{\overline{T}^j_{ij}\quad1}_{v_j}\quad \allzero^T \underbrace{-1}_{\gamma_k}\quad \allzero^T\Bigg].
\end{align}
Since the identity matrix $I_m$ is full rank and $R_c(\varphi)=S(\varphi)$, we have $R_c(\varphi)\delta\varphi=0$ if and only if there exists $\delta\zeta =[\delta\varphi^T,\delta\gamma^T]^T$ such that $S''(\zeta)\delta\zeta=0$.

\end{proof}

It can be observed that the left hand side of (\ref{bearing_constr}) is the expression of the bearing  between the point $(\alpha_i,\beta_i)$ (or $(\alpha_j,\beta_j)$) and $(0,\gamma_k)$ and the right hand side is a constant. Therefore, the Jacobian $S''(\zeta)$ actually follows the definition of bearing rigidity matrix in a new clock framework, whose graph is constructed from the original graph and reflects the underlying bearing property of the clock constraints. We next define this new clock framework and explore how it relates to the original framework.

Define a new graph $\mathcal{G}'=(\mathcal{V}\cup\mathcal{V}',\mathcal{E}_1\cup\mathcal{E}_2)$, where $\mathcal{V}$ is constructed according to the edge in $\mathcal{G}$, i.e., $\mathcal{V}'=\{v'_1,...,v'_m\}$ and $v'_k$ corresponds to the $k$th edge $\{v_i,v_j\}\in\mathcal{E}$. The edge set $\mathcal{E}_1$ is constructed so that for every $\{v_i,v_j\}\in\mathcal{E}$, we have $\{v_i,v'_k\},\{v_j,v'_k\}\in\mathcal{E}_1$. The edge set $\mathcal{E}_2$ is constructed so that $(\mathcal{V}',\mathcal{E}_2)$ is a spanning tree of the line graph of $\mathcal{G}$, i.e. the adjacent vertices in $(\mathcal{V}',\mathcal{E}_2)$ imply the adjacent edges in $\mathcal{G}$. The line graph of any connected graph is connected  \cite{chartrand1969connectivity}, so there always exists a spanning tree. We denote $\eta=[\eta_1,...\eta_m]^T$ where $\eta_k=[0,\gamma_k]^T$ and define a clock framework $(\mathcal{G}',[\varphi^T,\eta^T]^T)$. So the configuration $[\varphi^T,\eta^T]^T$ provides a mapping from vertex $v_i\in\mathcal{V}$ to the points $\varphi_i$ and from $v'_k\in\mathcal{V}'$ to the points $\eta_k$. Given $(d_{in},T_s)$, $(\mathcal{G}',[\varphi^T,\eta^T]^T)$ can be uniquely decided by $(\mathcal{G},\varphi)$.

\begin{figure}
  \begin{center}
  \includegraphics[scale=0.27]{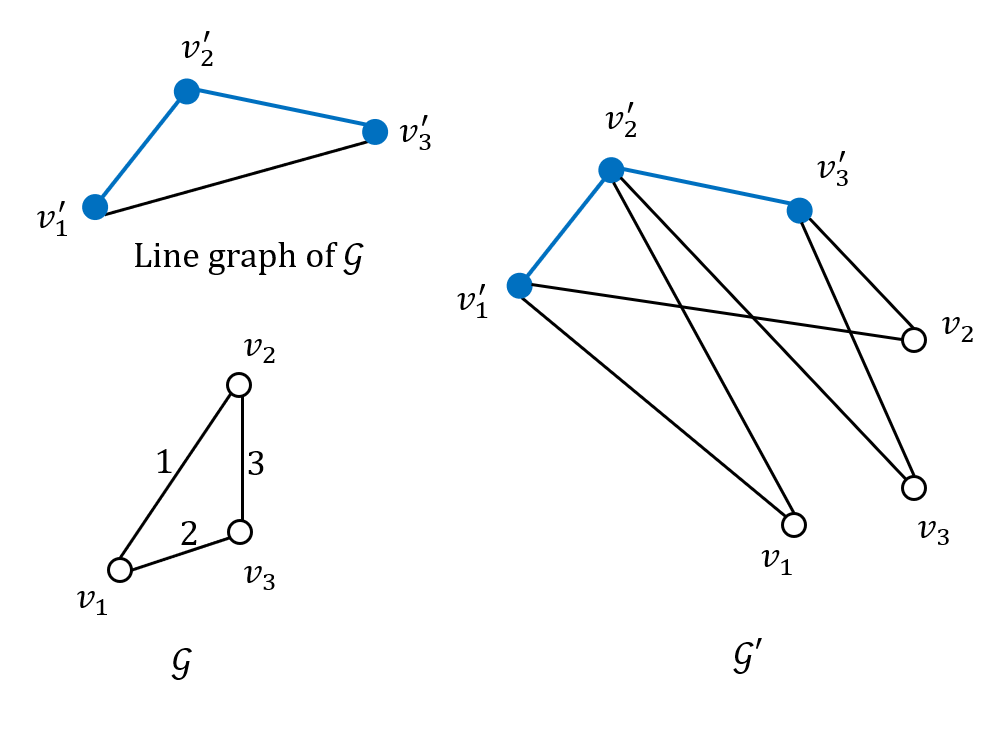}\\
  \caption{Original graph $\mathcal{G}$ to new graph $\mathcal{G}'$. The solid blue vertices correspond to $\mathcal{V}'$ and the blue edges correspond to $\mathcal{E}_2$, which is a spanning tree of the line graph of $\mathcal{G}$.}
  \label{fig_GtoG'}
  \end{center}
\end{figure}
\begin{figure}
  \begin{center}
  \includegraphics[scale=0.18]{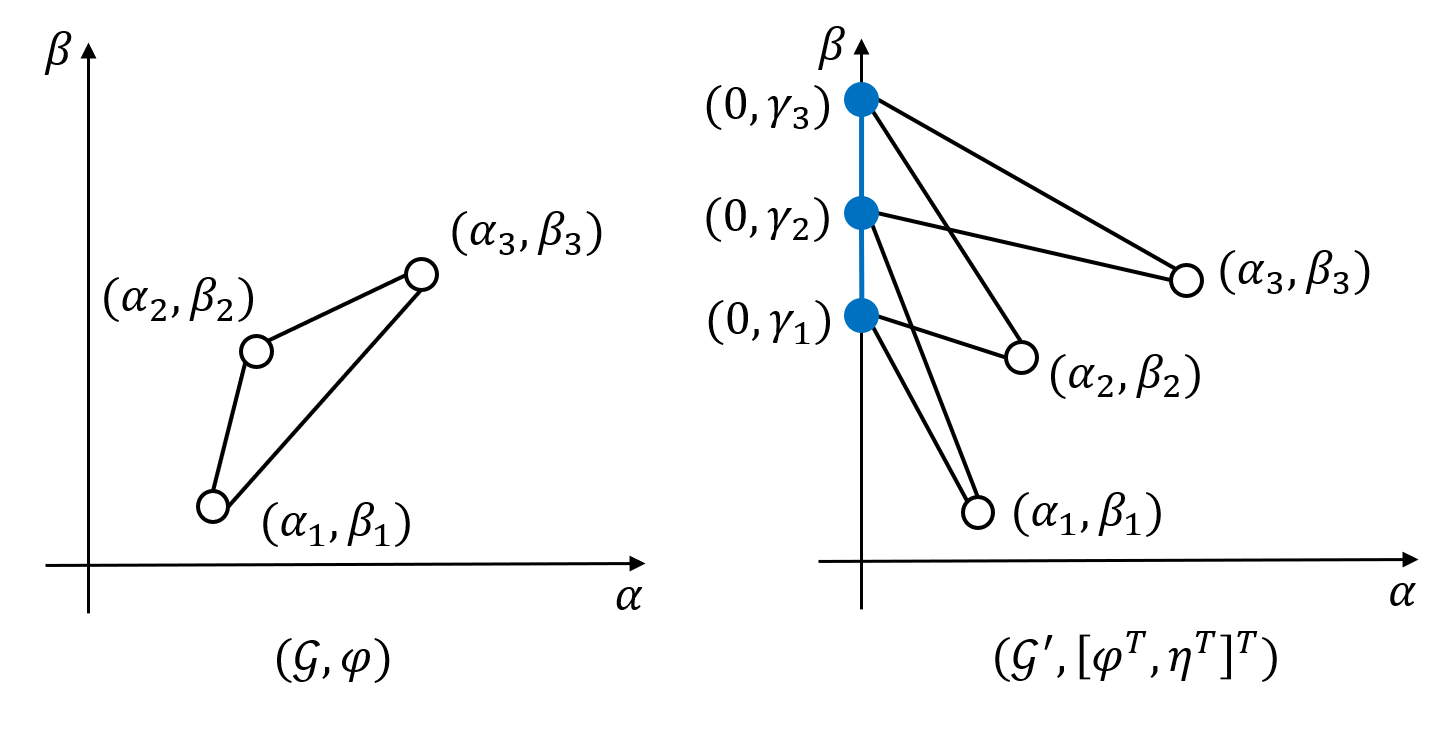}\\
  \caption{Original clock framework $(\mathcal{G},\varphi)$ and new clock framework $(\mathcal{G}',[\varphi^T,\eta^T]^T)$. The solid blue vertices correspond to $\mathcal{V}'$ and the blue edges correspond to $\mathcal{E}_2$.}
  \label{fig_twoframework}
  \end{center}
\end{figure}

Now we recall some concepts from bearing rigidity theory. A framework $(\mathcal{G},p)$ is a combination of graph $\mathcal{G}$ and a configuration in $d$-dimensional space $p=[p_1^T,...,p_n^T]$ where $p_i\in\mathbb{R}^d$. A framework $(\mathcal{G},p)$ is infinitesimally bearing rigid if all the infinitesimal motions are trivial, i.e., the same up to translation and scaling \cite{zhao2015bearing}. 

\begin{defn}[\!\!\cite{zhao2017laman}]\label{def_GBR}
 A graph $\mathcal{G}$ is \textbf{generically bearing rigid} in $\mathbb{R}^d$ if there exists at least one configuration $p$ in $\mathbb{R}^d$ such that $(\mathcal{G},p)$ is infinitesimally bearing rigid.
\end{defn}
\begin{lem}[\!\!\cite{zhao2017laman}]
 If $\mathcal{G}$ is generically bearing rigid in $\mathbb{R}^d$, then $(\mathcal{G},p)$ is infinitesimally bearing rigid for almost all $p$ in $\mathbb{R}^d$.
\end{lem} 

Now we are ready to prove the relation between graph $\mathcal{G}$ and $\mathcal{G}'$. First we give the definition of redundant edge.
\begin{defn}
An edge in a generically rigid graph $\mathcal{G}$ is a \textbf{redundant edge} if $\mathcal{G}$ is still generically rigid after removing this edge.
\end{defn}

The following lemma gives the necessary and sufficient condition between the bearing rigidity of graph $\mathcal{G}$ and $\mathcal{G'}$.

\begin{lem}\label{lem_GtoG'}
The graph $\mathcal{G}'$ is generically bearing rigid in $\mathbb{R}^2$ if and only if $\mathcal{G}$ is generically bearing rigid in $\mathbb{R}^2$ with at least one redundant edge.
\end{lem}

\begin{proof}
The Henneberg operation is a useful method for adding new vertices into a graph while preserving its rigidity. For a graph $\mathcal{G}=(\mathcal{V},\mathcal{E})$, the vertex addition operation in $\mathbb{R}^d$ adds a new vertex with $d$ new incident edges. The edge splitting operation in $\mathbb{R}^d$ removes an existing edge and adds a new vertex with $d+1$ new incident edges. It is known that both operations preserve the bearing rigidity of graphs in $\mathbb{R}^d$ \cite{zhao2017laman}.

Suppose that $\mathcal{G}$ is generically bearing rigid. We denote a spanning tree of the line graph of $\mathcal{G}$ as $\mathcal{G}_{s}$, so the $k$th vertex of $\mathcal{G}_{s}$ corresponds to the $k$th edge of $\mathcal{G}$. Define a set $\mathcal{M}=\emptyset$. Then we do the following steps to get a graph $\mathcal{G}'$ from $\mathcal{G}$:

\textbf{Step 1}: Remove a redundant edge $k_1$ corresponding to $\{v_p,v_q\}$ in $\mathcal{G}$. By definition, the resulting graph is still generically bearing rigid. Then, add a new vertex $v'_{k_1}$ with two new edges $\{v_p,v'_{k_1}\}$ and $\{v_q,v'_{k_1}\}$ to the graph $\mathcal{G}$. Add the integer $k_1$ into the set $\mathcal{M}$, i.e.,  $\mathcal{M}=\{k_1\}$. \textbf{(vertex addition)}

\textbf{Step 2}: Select a remaining edge $k$ corresponding to $\{v_i,v_j\}\in \mathcal{E}$ where the $k$th vertex in $\mathcal{G}_{s}$ is adjacent to the $k_\ell$th vertex and $k_\ell\in\mathcal{M}$. Remove the edge $k$ in $\mathcal{G}$ and add a new vertex $v'_k$.  Then add three new edges $\{v_i,v'_{k}\},\{v_j,v'_k\},\{v'_{k},v'_{k_\ell}\}$, where the first two edge are in $\mathcal{E}_1$ and the last one is in $\mathcal{E}_2$. Add the integer $k$ into the set $\mathcal{M}$. Repeat this step for the remaining edges in $\mathcal{E}$ until all the edges in $\mathcal{E}$ are removed. \textbf{(edge splitting)}

Since all of the above steps are rigidity-preserving, $\mathcal{G}'$ is generically bearing rigid. The corresponding deconstruction steps are also rigidity-preserving, so the proof is reversible.
\end{proof}

\begin{figure}[t]
  \begin{center}
  \includegraphics[scale=0.3]{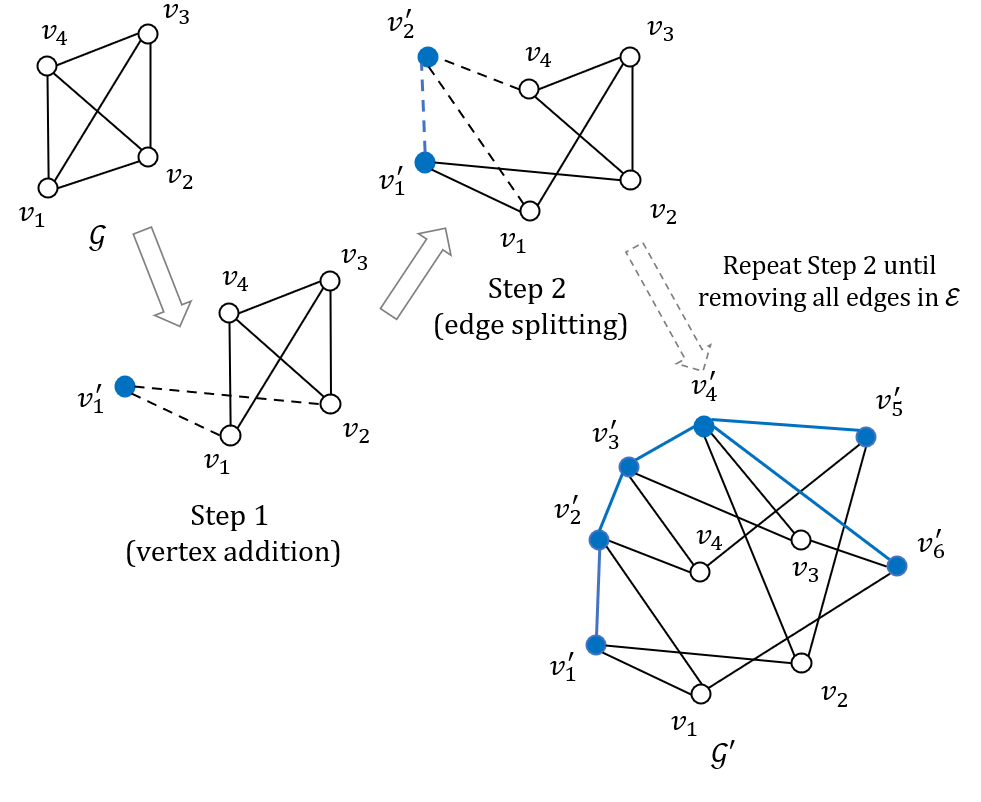}\\
  \caption{Illustration of Henneberg operation in Lemma \ref{lem_GtoG'}. Graph $\mathcal{G}$ is generically bearing rigid in $\mathbb{R}^2$ with one redundant edge, so the resulting graph $\mathcal{G}'$ is generically bearing rigid. Dashed blue edges in $\mathcal{G}'$ correspond to $\mathcal{E}_2$.}
  \label{fig_Henneberg}
  \end{center}
\end{figure}

Lemma \ref{lem_GtoG'} establishes the rigidity relation between the original graph $\mathcal{G}$ and the new graph $\mathcal{G}'$. Recall, the Jacobian of the left hand side of (\ref{bearing_constr}) with respect to $\zeta$ is $S''(\zeta)$. Next we will prove the relation between the bearing rigidity of the new clock framework $(\mathcal{G},[\varphi^T,\eta^T]^T)$ and the rank of the matrix $S''(\zeta)$.

Before proceeding, we need to introduce Assumption \ref{assum2} (Device assumption), which can be trivially satisfied in a typical UWB implementation.

\begin{subtheorem}{assumption} \label{assum2}
\begin{assumption}[Single antenna]\label{assum_antenna}
A UWB sensor has only a single antenna, so it can not receive more than one message at a time, i.e., there exists no $T^i_{(a,i)}=T^i_{(b,i)}$ for all $a\neq b$.
\end{assumption}
\begin{assumption}[Broadcast scheme] \label{assum_broadcast}
Every node in the UWB sensor network broadcasts once in one round of measurements, i.e., the sending timestamps $T_{(i,j)}^i$ at node $i$ are the same for all node $j$ in the neighborhood of node $i$.
\end{assumption}
\end{subtheorem}

Assumption \ref{assum_antenna} is satisfied by the typical implementation of UWB sensors and Assumption \ref{assum_broadcast} can be trivially satisfied by the communication protocol design of the UWB sensor networks, which also reduces the communication complexity of the network. Assumption \ref{assum2} provides a useful inequality as shown in Lemma \ref{lem_broadcast}.

\begin{lem}\label{lem_broadcast}
 Under Assumption \ref{assum2}, for a graph $\mathcal{G}'=(\mathcal{V}\cup\mathcal{V}',\mathcal{E}_1\cup\mathcal{E}_2)$ defined above, $\gamma_k\neq\gamma_{\ell}$ for any $\{v'_k,v'_\ell\}\in\mathcal{E}_2$.
\end{lem}
\begin{proof}
Consider an edge $\{v'_k,v'_\ell\}\in\mathcal{E}_2$ where $v'_k$  and $v'_{\ell}$ correspond to edges $\{v_i,v_a\}\in\mathcal{E}$ and $\{v_i,v_b\}\in\mathcal{E}$, respectively. According to (\ref{seperate_constr}), 
\begin{subequations}
\begin{align}
    \gamma_k&=\overline{T}^i_{ia}\alpha_i+\beta_i=\frac{T^i_{(i,a)}+T^i_{(a,i)}}{2}\alpha_i+\beta_i\\
    \gamma_{\ell}&=\overline{T}^i_{ib}\alpha_i+\beta_i=\frac{T^i_{(i,b)}+T^i_{(b,i)}}{2}\alpha_i+\beta_i,
\end{align}
\end{subequations}
where $T^i_{(a,i)}\neq T^i_{(b,i)}$ due to the Assumption \ref{assum_antenna} and $T^i_{(i,b)}=T^i_{(i,a)}$ due to the Assumption \ref{assum_broadcast}. So $\gamma_k\neq\gamma_{\ell}$ for any $\{v'_k,v'_\ell\}\in\mathcal{E}_2$.
\end{proof}

Lemma \ref{lem_broadcast} plays an important role in the proof of the following lemma. Under Assumption \ref{assum2}, Lemma \ref{lem_IBRtoS''} establishes the relation between the bearing rigidity of $(\mathcal{G'},[\varphi^T,\eta^T]^T)$ and the rank of matrix $S''(\zeta)$. Note that Lemma \ref{lem_IBRtoS''} still holds without Assumption \ref{assum2}, but only for a subset of the generic set of configurations with which $\gamma_k\neq\gamma_{\ell}$ for all $\{v'_k,v'_\ell\}\in\mathcal{E}_2$. The necessity proof follows a similar structure from \cite{shames2009minimization}.

\begin{lem}\label{lem_IBRtoS''}
Under Assumption \ref{assum2}, the clock framework $(\mathcal{G'},[\varphi^T,\eta^T]^T)$ is infinitesimally bearing rigid if and only if $\rank(S''(\zeta))=2n+m-2$ and $\Null(S''(\zeta))=\Span{[\varphi^T,\gamma^T]^T,[\allone_n^T\otimes[0,1],\allone_m^T]^T}$.
\end{lem}

\begin{proof}
Denote the bearing rigidity matrix of the clock framework $(\mathcal{G}',[\varphi^T,\eta^T]^T)$ as $R_b$. Then after permutation,
\begin{align}\label{PermutedRb}
    P_1R_bP_2=\begin{bmatrix}
        S''(\zeta) & A\\
        \allzero & B
    \end{bmatrix}\in\mathbb{R}^{2(3m-1)\times 2(n+m)},
\end{align}
where $A$ is a submatrix of $R_b$, whose columns correspond to the x-coordinate of points in $\mathcal{V}'$ and rows correspond to the edges in $\mathcal{E}_1$. The matrix  $B=QH\otimes[1,0]^T\in\mathbb{R}^{2(m-1)\times m}$ where $Q=\diag{|\gamma_k-\gamma_\ell|^{-1}}$ for all  $\{v'_k,v'_\ell\}\in\mathcal{E}_2$ and $H$ is the incidence matrix of $(\mathcal{V}',\mathcal{E}_2)$. Lemma \ref{lem_broadcast} guarantees that $\gamma_k\neq\gamma_{\ell}$ in the matrix $Q$. The matrix $P_1$ permutes rows of $R_b$ so that the first $2m$ rows correspond to the edges in $\mathcal{E}_1$ and the last $2(m-1)$ rows correspond to the edges in $\mathcal{E}_2$. The matrix $P_2$ permutes columns of $R_b$ so that the first $2n+m$ columns correspond to $[\varphi^T,\gamma^T]^T$, i.e., the coordinates of points in $\mathcal{V}$ and the y-coordinate of points in $\mathcal{V}'$. The last $m$ columns after permutation correspond to the x-coordinate of points in $\mathcal{V}'$.

(Necessity) Suppose that $(\mathcal{G'},[\varphi^T,\eta^T]^T)$ is infinitesimally bearing rigid. The permuted bearing rigidity matrix $P_1R_bP_2$ has same rank as $R_b$, i.e., $\rank(P_1R_bP_2)=2(n+m)-3$ and $\Null(P_1R_bP_2)=\Span{u_1,u_2,u_3}$ where $u_1=[\varphi^T,\gamma^T,\allzero^T]^T$, $u_2=[\allone_n^T\otimes [0,1],\allone_m^T,\allzero^T]^T$ ,and $u_3=[\allone_n^T\otimes [1,0],\allzero^T,\allone_m^T]^T$ \cite{zhao2015bearing}.

Suppose that a vector $w\in\Null(S'')$ and let $u=[w^T,\allzero^T]^T\in\mathbb{R}^{2(n+m)}$ then $u\in \Null(P_1R_bP_2)$. Hence $u$ must be a linear combination of $u_1$, $u_2$, and $u_3$, i.e., for some scalar $a$, $b$, $c$, not all zero,
\begin{align*}
    u=au_1+bu_2+cu_3,
\end{align*}
and the last $m$ rows gives
\begin{align*}
    \allzero=a\bar{u}_1+b\bar{u}_2+c\bar{u}_3,
\end{align*}
where $\bar{u}_1=\allzero$, $\bar{u}_2=\allzero$ and, $\bar{u}_3=\allone_m$. 
So, $c=0$ and $a,b$ are not both zero. For the remaining $2n+m$ rows,
\begin{align}
    w=au'_1+bu'_2,
    \label{eqn_nullRr}
\end{align}
where $u'_1=[\varphi^T,\gamma^T]^T$ and $u'_2=[\allone_n^T\otimes[0,1],\allone_m^T]^T$.
Since (\ref{eqn_nullRr}) holds for any $w\in \Null(S'')$, $\rank(S'')=2n+m-2$ and $\Null(S'')=\Span{u'_1,u'_2}$. 

(Sufficiency) Suppose that $\Null(S'')=\Span{u'_1,u'_2}$. The matrix $B=QH\otimes[1,0]^T$ where $Q$ is a diagonal matrix and $H$ is the incidence matrix of $(\mathcal{V}',\mathcal{E}_2)$, hence $\Null(B)=\Span{\allone_m}$. Thus $\Null(P_1R_b P_2)=\Span{[{u'_1}^T,\allzero^T]^T,[{u'_2}^T,\allzero^T]^T,[{u'_3}^T,\allone_m^T]^T}$, where $S''u'_3+A\allone_m=\allzero$. It follows immediately that $\Null(P_1R_b P_2)=\Span{u_1,u_2,u_3}$ and hence the clock framework
 $(\mathcal{G}',[\varphi^T,\eta^T]^T)$ is infinitesimally bearing rigid.
\end{proof}

The following theorem states the relation between the clock rigidity of clock framework $(\mathcal{G},\varphi)$ and the bearing rigidity of clock framework $(\mathcal{G}',[\varphi^T,\eta^T]^T)$.  

\begin{thm}\label{thm_IBR<->ICR}
Under Assumption \ref{assum2}, a clock framework $(\mathcal{G},\varphi)$ is infinitesimally clock rigid if and only if $(\mathcal{G}',[\varphi^T,\eta^T]^T)$ is infinitesimally bearing rigid.
\end{thm}
\begin{proof}
We first prove the necessity. Suppose $(\mathcal{G},\varphi)$ is infinitesimally clock rigid. By Theorem \ref{thm_ICR}, we have
\begin{align}
    \Null(R_c(\varphi))=\Span{\allone_n\otimes[0,1]^T,\varphi}.
\end{align}
So Lemma \ref{lem_RctoS''} gives that $\rank(S'')=2n+m-2$ and $\Null(S'')=\Span{[\allone_n^T\otimes[0,1],\allone_m^T]^T, [\varphi^T,\gamma^T]^T}$. By Lemma \ref{lem_IBRtoS''}, $(\mathcal{G}',[\varphi^T,\eta^T]^T)$ is infinitesimally bearing rigid. Since Theorem \ref{thm_ICR} and Lemmas \ref{lem_RctoS''} and \ref{lem_IBRtoS''} all state necessary and sufficient conditions,  the proof for the sufficiency of this Theorem also holds.
\end{proof}

\begin{cor}\label{cor_genericforIBR}
 Under Assumption \ref{assum2}, a clock configuration $[\varphi^T,\eta^T]^T$ is generic for bearing rigid of graph $\mathcal{G}'$ if and only if the clock configuration $\varphi$ is generic for clock rigidity of graph $\mathcal{G}$.
\end{cor}

Now we are ready to prove the following theorem, which gives a sufficient and necessary graph property for establishing infinitesimal clock rigidity. 

\begin{thm}[Main result]\label{thm_GBR<->ICR}
Under Assumption \ref{assum2}, for any generic clock configuration $\varphi$, a clock framework $(\mathcal{G},\varphi)$ is infinitesimally clock rigid if and only if $\mathcal{G}$ is generically bearing rigid in $\mathbb{R}^2$ with at least one redundant edge.
\end{thm}
\begin{proof}
 (Necessity) Suppose $(\mathcal{G},\varphi)$ is infinitesimally clock rigid. By Theorem \ref{thm_IBR<->ICR}, $(\mathcal{G}',[\varphi^T,\eta^T]^T)$ is infinitesimally bearing rigid and hence $\mathcal{G}'$ is generically bearing rigid. Then by Lemma \ref{lem_GtoG'}, $\mathcal{G}$ is generically bearing rigid in $\mathbb{R}^2$ with at least one redundant edge. 

(Sufficiency)
Suppose $\mathcal{G}$ is generically bearing rigid in $\mathbb{R}^2$ with at least one redundant edge. By Lemma \ref{lem_GtoG'}, $\mathcal{G}'$ is generically bearing rigid in $\mathbb{R}^2$. Since $\varphi$ is a generic configuration for $\mathcal{G}$, by Corollary \ref{cor_genericforIBR}, $(\mathcal{G}',[\varphi^T,\eta^T]^T)$ is infinitesimally bearing rigid. By Theorem \ref{thm_IBR<->ICR}, $(\mathcal{G},\varphi)$ is infinitesimally clock rigid.

The theorems and lemmas contributing to this proof are shown in Fig. \ref{fig_thm_dependency}.
\end{proof}

\begin{figure}
  \begin{center}
  \includegraphics[scale=0.25]{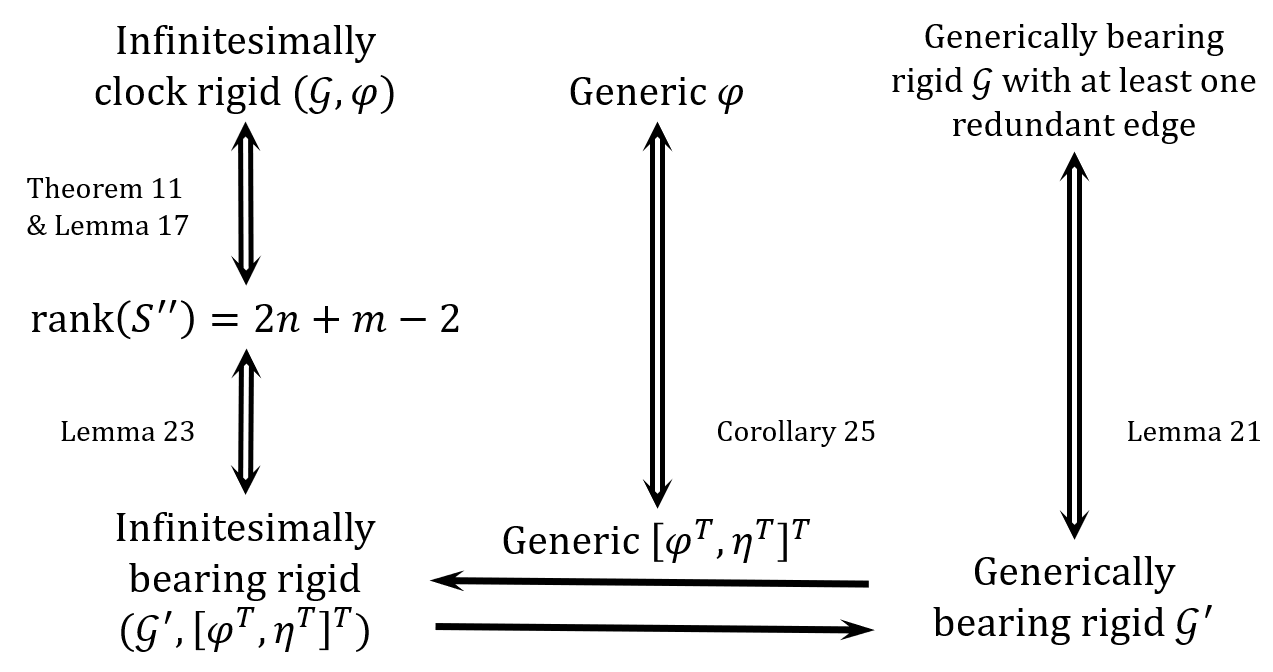}\\
  \caption{Diagram of the statements contributing to the proof of Theorem \ref{thm_GBR<->ICR}.}
  \label{fig_thm_dependency}
  \end{center}
\end{figure}

Theorem \ref{thm_GBR<->ICR} suggests that the infinitesimal clock rigidity can be determined by checking if the underlying graph is generically bearing rigid with at least one redundant edge. Up to this point there is no existing direct graph-based method to check clock rigidity. Theorem \ref{thm_GBR<->ICR} can be used to provide a topological method to establish infinitesimal clock rigidity based on Laman graphs; defined here. 
\begin{defn}[\!\!\cite{anderson2008rigid,tay1985generating}]
A graph $\mathcal{G}=(\mathcal{V},\mathcal{E})$ is \textbf{Laman} if $|\mathcal{E}|=2|\mathcal{V}|-3$ and every subset of $k\geq 2$ vertices spans at most $2k-3$ edges.
\end{defn}

The topological result is stated in the subsequent corollary and draws on the following theorem.
\begin{thm}[\!\!\cite{zhao2017laman}]\label{thm_GBR_Laman}
A graph $\mathcal{G}$ is generically bearing rigid in $\mathbb{R}^2$ if and only if the graph contains a Laman spanning subgraph.
\end{thm}
\begin{cor}\label{cor_ICR_Laman}
Under Assumption \ref{assum2}, for any generic clock configuration $\varphi$, a clock framework $(\mathcal{G},\varphi)$ is infinitesimally clock rigid if and only if $\mathcal{G}$ contains a Laman spanning subgraph $\mathcal{G}_\ell$ and $\mathcal{G}\neq\mathcal{G}_\ell$.
\end{cor}
\begin{proof}
The result follows directly from the statements of Theorem \ref{thm_GBR<->ICR} and Theorem \ref{thm_GBR_Laman}.
\end{proof}

The following example exercises Corollary \ref{cor_ICR_Laman} for the clock framework $(K_3,\varphi)$ and $(K_4,\varphi)$, respectively.

\begin{exmp}\label{exmp_K3andK4}
The complete graph $K_3$ is generically bearing rigid in $\mathbb{R}^2$ since it contains a Laman spanning subgraph $\mathcal{G}_\ell$ as shown in Fig. \ref{fig_K3andK4}, but $\mathcal{G}_\ell=\mathcal{G}$, i.e., there is no redundant edge for its rigidity. Following from Corollary \ref{cor_ICR_Laman}, $(K_3,\varphi)$ is not infinitesimally clock rigid for any generic configuration $\varphi$.

The complete graph $K_4$ contains a Laman spanning subgraph $\mathcal{G}_\ell$ and $\mathcal{G}_\ell\neq\mathcal{G}$, so $(K_4,\varphi)$ is infinitesimally clock rigid following from Corollary \ref{cor_ICR_Laman}. Note that $K_4$ is also the minimal graph which establishes infinitesimal clock rigidity property.
\end{exmp}

\begin{figure}[!t]
\begin{center}
\subfigure[Non-infinitesimally clock rigid]{\includegraphics[scale=0.4]{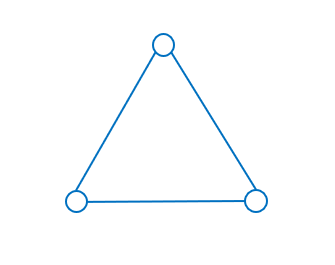}%
\label{Fig_K3}}
\hfil
\subfigure[Infinitesimally clock rigid]{\includegraphics[scale=0.4]{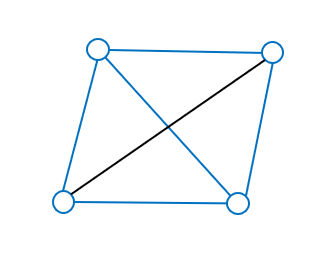}%
\label{Fig_K4}}
\end{center}
\caption{Complete graph $K_3$ and $K_4$. The blue vertices and edges show the Laman spanning subgraph. The black edge in $K_4$ is the redundant edge for its bearing rigidity.}
 \label{fig_K3andK4}
\end{figure}

Theorem \ref{thm_GBR<->ICR} and Corollary \ref{cor_ICR_Laman} also show great value in the joint position and clock problem, which is further discussed in Sections \ref{sec_jointmatrix} and \ref{sec_estimation}.

\section{Joint Rigidity} \label{sec_jointmatrix}
It is studied in clock rigidity theory whether a clock framework with certain graph property can be determined up to some trivial variations given the TOA timestamp measurements between neighbors. The close relation between distance and time in TOA measurements also provides an invariant equality involving both position and clock information. In this section, we combine clock rigidity theory and distance rigidity theory to analyze the joint position and clock problem. We explore the conditions under which the position and clock information in a framework can be uniquely and simultaneously determined up to some trivial variations.

Consider a TOA-based UWB sensor network. Define a position configuration $p=[p_1^T,...,p_n^T]^T\in\mathbb{R}^{nd}$ and a clock configuration $\varphi=[\varphi_1^T,...,\varphi_n^T]^T\in\mathbb{R}^{2n}$. We next take both position and clock into account and define a position-clock framework $(\mathcal{D},\sigma)$ where $\mathcal{D}=(\mathcal{V},\mathcal{E}_\mathcal{D})$ is a directed graph and $\sigma=[p^T,\varphi^T]^T\in\mathbb{R}^{n(d+2)}$ is a position-clock configuration. The position-clock configuration provides a mapping from  $v_i\in\mathcal{V}$ to $\sigma_i=[p^T_i,\varphi^T_i]^T\in\mathbb{R}^{d+2}$, including a position $p_i\in\mathbb{R}^d$ and a clock $\varphi_i\in\mathbb{R}^2$. 

Based on (\ref{OWRdist}), the position-clock framework satisfies the following constraint for every edge $(v_i,v_j)\in\mathcal{E}_\mathcal{D}$
\begin{align}\label{constr_tof}
   \norm{p_i-p_j}^2-c^2(\alpha_jT^j_{(i,j)}+\beta_j-\alpha_iT^i_{(i,j)}-\beta_i)^2=0,
 \end{align}
 where $c$ is the speed of light. We take the partial derivative of the left hand side of (\ref{constr_tof}) with respect to $\sigma$ for every $(v_i,v_j)\in\mathcal{E}_\mathcal{D}$ and call the resulting Jacobian $R_{cd}(\sigma)$ with respect to $\sigma$ the joint rigidity matrix since it includes the information about both clock rigidity and distance rigidity. Each row of $R_{cd}(\sigma)$ corresponds to an edge $(v_i,v_j)$ which has the form
 \begin{align}\label{joint_constr}
    \Bigg[&\allzero^T\quad\underbrace{(p_i-p_j)^T}_{v_i}\quad\allzero^T\quad\underbrace{(p_j-p_i)^T}_{v_j}\quad \allzero^T\quad   \underbrace{cd_{ij}T^i_{(i,j)}\quad cd_{ij}}_{v_i}\nonumber\\
     &\allzero^T\quad\underbrace{-cd_{ij}T^j_{(i,j)}\quad-cd_{ij}}_{v_j}\quad\allzero^T\Bigg]\in\mathbb{R}^{(d+2)n},
\end{align}
where $d_{ij}=c(\alpha_jT^j_{(i,j)}+\beta_j-\alpha_iT^i_{(i,j)}-\beta_i)$.

Let $\delta\sigma$ be a joint variation of the configuration $\sigma$. If $R_{cd}(\sigma)\delta\sigma=0$, then $\delta\sigma$ is called an infinitesimal joint variation of $(\mathcal{D},\sigma)$. Infinitesimal joint variations preserve timestamp measurements. An infinitesimal joint variation is called trivial if it corresponds to a translation and a rotation of position configuration $p$, a translation of the clock offset configuration $\beta$ and a scaling of the entire position-clock framework. See Fig. \ref{fig_joint_trivial}. Analogous to infinitesimal distance (bearing) rigidity, we define infinitesimal joint rigidity.
 
 \begin{defn}\label{def_IJR}
 A position-clock framework is \textbf{infinitesimally joint rigid} if all the infinitesimal joint variations are trivial.
 \end{defn}
 
 We next give a necessary and sufficient condition of infinitesimal joint rigidity.
 
\begin{figure}[!t]
\begin{center}
\subfigure[Translation and rotation of position configuration $p$]{\includegraphics[scale=0.24]{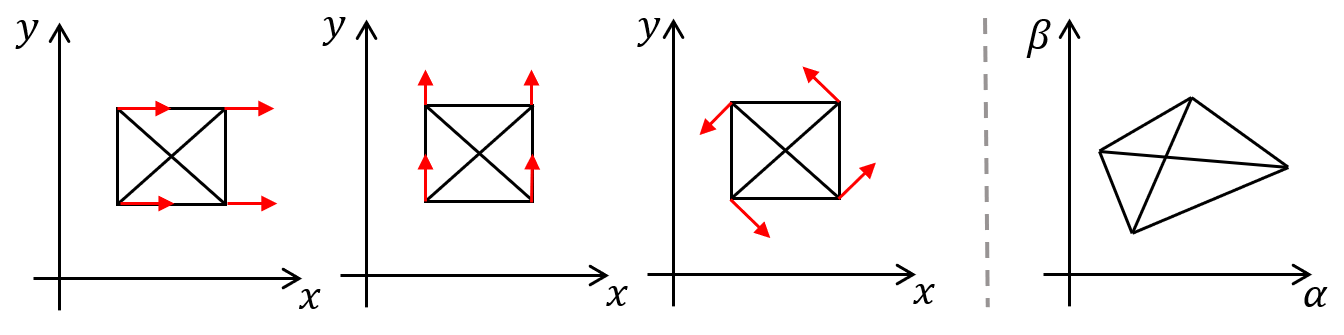}%
\label{Fig_joint_trivial_a}}
\hfil
\subfigure[Translation of $\beta$ ]{\includegraphics[scale=0.24]{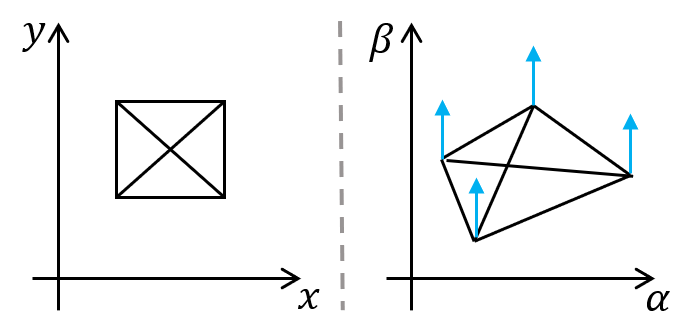}%
\label{Fig_joint_trivial_b}}
\hfil
\subfigure[Scaling of the entire framework]{\includegraphics[scale=0.24]{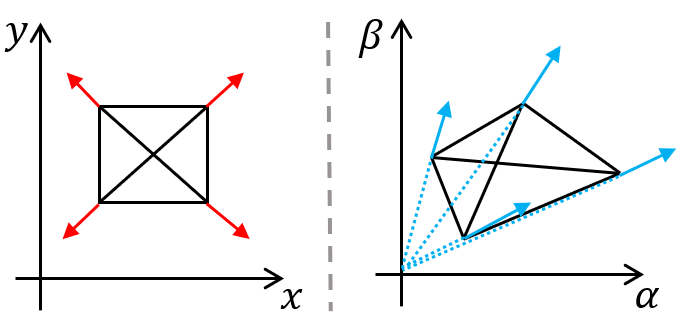}%
\label{Fig_joint_trivial_c}}
\end{center}
\caption{Basic trivial infinitesimal joint variations in $\mathbb{R}^{2+2}$. Position variations (red arrows) and clock variations (blue arrows) are shown in the $xy$-coordinate plane  and the $\alpha\beta$-coordinate plane, respectively.}
  \label{fig_joint_trivial}
\end{figure}

\begin{thm}\label{thm_IJRrank}
A position-clock framework $(\mathcal{D},\sigma)$ in $\mathbb{R}^{d+2}$ is infinitesimally joint rigid if and only if $\rank(R_{cd}(\sigma))=(d+2)n-d(d+1)/2-2$.
\end{thm}
\begin{proof}
The expression of joint rigid matrix $R_{cd}$ in (\ref{joint_constr}) shows that  $\Span{[(\allone_n\otimes I_d)^T,\allzero^T]^T, [p^T(I_n\otimes J_d^1)^T,\allzero^T]^T,...,[p^T(I_n\otimes J_d^{d(d-1)/2})^T,\allzero^T]^T}\subseteq \Null(R_{cd})$. It can be observed that these vectors correspond to translations and rotations of position configuration $p$ in $d$-dimensional space. We also have $[\allzero^T,\allone_n^T\otimes [0,1]]^T\subseteq\Null(R_{cd})$, corresponding to a translation of the clock offset configuration $\beta$.

It is also true that $\Span{\sigma}\subseteq\Null(R_{cd})$ since $\sigma=[p^T,\varphi^T]^T$ and $R_{cd}(\sigma)\sigma=0$ satisfies the constraints in (\ref{constr_tof}). The configuration $\sigma$ corresponds to a scaling of the position-clock framework. Following from Definition \ref{def_IJR}, $(\mathcal{D},\sigma)$ is infinitesimally joint rigid if and only if $\Dim(\Null(R_{cd}))=d(d+1)/2+2$, i.e., $\rank(R_{cd}(\sigma))=(d+2)n-d(d+1)/2-2$.
\end{proof}

A consequence of Theorem \ref{thm_IJRrank} is a graph $\mathcal{D}$ must be sufficiently connected to be infinitesimally joint rigid.

\begin{cor}\label{cor_minimumedges}
An infinitesimally joint rigid position-clock framework $(\mathcal{D},\sigma)$ in $\mathbb{R}^{d+2}$ only if it has at least $(d+2)n-d(d+1)/2-2$ edges.
\end{cor}

Fig. \ref{fig_digraphs} shows some examples of infinitesimal joint rigidity and flexibility in $\mathbb{R}^{2+2}$. Fig. \ref{fig_digraphs}(a) and \ref{fig_digraphs}(b) are infinitesimally joint rigid since the rank of their joint rigidity matrix equals to $4n-5$. They also exhibit the minimum number of directed edges for joint rigidity on $5$ and $6$ node graphs. Fig. \ref{fig_digraphs}(c) has the same number of edges as Fig. \ref{fig_digraphs}(b) but it is non-infinitesimally joint rigid, showing that the necessary condition in Corollary \ref{cor_minimumedges} is not sufficient.

\begin{figure}[!t]
\begin{center}
\subfigure[Infinitesimally joint rigid ($5$ vertices, $15$ edges, $\rank(R_{cd})=15$)]{\includegraphics[scale=0.265]{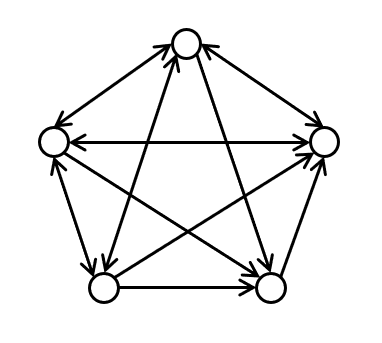}%
\label{Fig_digraph_rk15}}
\hfil
\subfigure[Infinitesimally joint rigid ($6$ vertices, $19$ edges, $\rank(R_{cd})=19$)]{\includegraphics[scale=0.265]{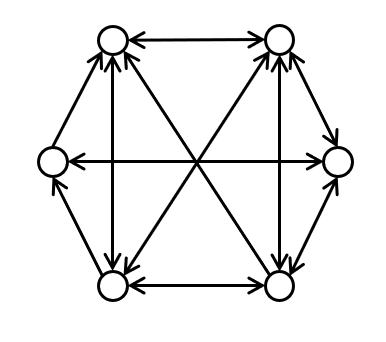}%
\label{Fig_digraph_rk19}}
\hfil
\subfigure[Non-infinitesimally joint rigid ($6$ vertices, $19$ edges, $\rank(R_{cd})=18$)]{\includegraphics[scale=0.265]{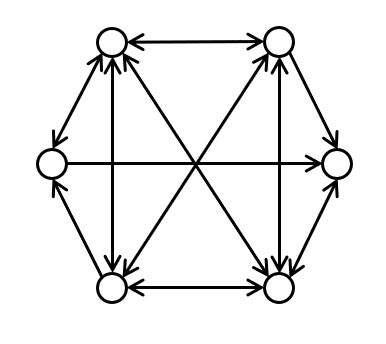}%
\label{Fig_digraph_rk18}}
\end{center}
\caption{Examples of position-clock frameworks in $\mathbb{R}^{2+2}$. }
  \label{fig_digraphs}
\end{figure}

We denote the disoriented graph of $\mathcal{D}=(\mathcal{V},\mathcal{E}_\mathcal{D})$ as $\mathcal{G}=(\mathcal{V},\mathcal{E})$, which is the undirected graph obtained after removing the orientation of the directed edges
of $\mathcal{D}$. In a directed graph $\mathcal{D}$, it is challenging to conclude generic graph properties from the joint rigidity matrix $R_{cd}(\sigma)$ due to its nontrivial expression. As an example, Fig. \ref{fig_digraphs}(b) and \ref{fig_digraphs}(c) show that even with the same disoriented graph and same number of edges, the joint rigidity of the position-clock framework is indeterminable.  We next leverage Assumption 1 (Bidirectional communication) to simplify the problem from a directed graph to an undirected graph and study the graph property for realizing infinitesimal joint rigidity.

Under Assumption 1, assume that $|\mathcal{E}_\mathcal{D}|=2m$, so $|\mathcal{E}|=m$. We reference the edge in $\mathcal{E}$ by edge index $k$ rather than node pair $\{v_i,v_j\}$, so  $d_{k}\triangleq d_{ij}=c(\alpha_jT^j_{(i,j)}+\beta_j-\alpha_iT^i_{(i,j)}-\beta_i)$ for all $\{v_i,v_j\}\in\mathcal{E}$ and $k=1,2,...,m$. The next lemma shows the characteristics of the joint rigidity matrix under Assumption 1. 

\begin{lem}\label{lem_TRcd}
 Given a position-clock framework $(\mathcal{D},\sigma)$ and its disoriented graph $\mathcal{G}$ in $\mathbb{R}^{d+2}$ under Assumption \ref{assump1}, the joint rigidity matrix $R_{cd}(\sigma)$ has the following form after elementary row operations 
 \begin{align}
    TR_{cd}(\sigma)=\begin{bmatrix}
    &R_d(p) &Y\\
 &\allzero &2cDR_c(\varphi)
 \end{bmatrix},
 \end{align}
 where $T$ represents the corresponding elementary row operations, $R_d(p)$ is the distance rigidity matrix of $(\mathcal{G},p)$, $R_c(\varphi)$ is the clock rigidity matrix of $(\mathcal{G},\varphi)$, $D=\diag{d_{k}}$ and $Y$ is a submatrix of $R_{cd}$, whose columns correspond to $\varphi$ and rows correspond to one of the directed edges between each neighboring node pair.
 \end{lem}
\begin{proof}
Under Assumption 1, $(v_i,v_j)\in\mathcal{E}_\mathcal{D}$ with $\sigma=[p^T,\varphi^T]^T$ if and only if $(v_j,v_i)\in\mathcal{E}_\mathcal{D}$. The measured distances satisfy $d_{ij}=d_{ji}$, as discussed in Section \ref{sec_clockrigidity}. The elementary row operation is $T=T_3T_1$ where $T_1$ is a row switching elementary matrix so that the $k$th row and $(m+k)$th row of $T_1R_{cd}$ correspond to edges $(v_i,v_j)$ and $(v_j,v_i)$, respectively, and $T_3$ is a row addition elementary matrix which has the form
\begin{align*}
    T_3=\begin{bmatrix}
        I_m & \allzero\\
        -I_m & I_m
    \end{bmatrix}.
\end{align*}
The proof follows directly.
\end{proof}

The next theorem provides a sufficient condition for infinitesimal joint rigidity.
\begin{thm}\label{thm_IDR+ICR->IJR}
Under Assumption \ref{assump1}, a position-clock framework $(\mathcal{D},\sigma)$ in $\mathbb{R}^{d+2}$ with $\sigma=[p^T,\varphi^T]^T$ is infinitesimally joint rigid if the position framework $(\mathcal{G},p)$ in $\mathbb{R}^d$ is infinitesimally distance rigid and the clock framework $(\mathcal{G},\varphi)$ is infinitesimally clock rigid where $\mathcal{G}$ is the disoriented graph of $\mathcal{D}$.
\end{thm}
\begin{proof}
Suppose that a position framework $(\mathcal{G},p)$ in $\mathbb{R}^d$ is infinitesimally distance rigid and a clock framework $(\mathcal{G},\varphi)$ is infinitesimally clock rigid. By Theorem \ref{thm_ICR} and the distance rigidity theory in \cite{anderson2008rigid}, we have $\Null(R_d(p))=\Span{\allone_n\otimes I_d,(I_n\otimes J_d^1)p,...,(I_n\otimes J_d^{d(d-1)/2})p}$ and $\Null(R_c(\varphi))=\Span{\allone_n\otimes[0,1]^T,\varphi}$.

By Lemma \ref{lem_TRcd}, from the expression of $Y$ and (\ref{constr_tof}), we have $R_d(p)\allzero+Y(\allone_n\otimes [0,1]^T)=0$ and $R_d(p)p+Y\varphi=0$. So $\rank(R_{cd}(\sigma))=\rank(R_d(p))+\rank(R_c(\varphi))=(dn-d(d+1)/2)+(2n-2)=(d+2)n-d(d+1)/2-2$. By Theorem \ref{thm_IJRrank}, $(\mathcal{D},\sigma)$ in $\mathbb{R}^{d+2}$ is infinitesimally joint rigid.
\end{proof}

Theorem \ref{thm_IDR+ICR->IJR} shows that under Assumption \ref{assump1} the joint rigidity of the position-clock framework $(\mathcal{D},\sigma)$ can be decoupled into a distance rigidity problem on the position framework $(\mathcal{G},p)$ and a clock rigidity problem on the clock framework $(\mathcal{G},\varphi)$. Consequently, distance rigidity theory in \cite{anderson2008rigid} and the clock rigidity theory in Section \ref{sec_clockrigidity} can be applied to establish joint rigidity.

A sufficient and necessary condition holds for infinitesimal joint rigidity with $d=2$, i.e., the corresponding position framework is in $2$-dimensional space.

\begin{thm}\label{thm_joint2D_NS}
Under Assumptions \ref{assump1} and \ref{assum2}, for a position-clock framework $(\mathcal{D},\sigma)$ in $\mathbb{R}^{2+2}$ with $\sigma=[p^T,\varphi^T]^T$ and corresponding disoriented graph $\mathcal{G}$, the following statements are equivalent for generic position configuration $p$ and generic clock configuration $\varphi$:
\begin{enumerate}[label=(\alph*)]
    \item the position-clock framework $(\mathcal{D},\sigma)$ in $\mathbb{R}^{2+2}$ is infinitesimally joint rigid;
    \item the clock framework $(\mathcal{G},\varphi)$ in $\mathbb{R}^2$ is infinitesimally clock rigid;
    \item the position framework $(\mathcal{G},p)$ in $\mathbb{R}^2$ is infinitesimally distance rigid with at least one redundant edge.
\end{enumerate}
\end{thm}
\begin{proof}
Since infinitesimal distance rigidity is equivalent to infinitesimal bearing rigidity in $\mathbb{R}^2$, (c) implies that $\mathcal{G}$ is generically bearing rigid with at least one redundant edge. So by Theorem \ref{thm_GBR<->ICR}, (b) is equivalent to (c). 

With the equivalence between (b) and (c), it follows immediately from Theorem \ref{thm_IDR+ICR->IJR} that (c)$\Rightarrow$(a). We next prove that (a)$\Rightarrow$(c).

Suppose the position-clock framework $(\mathcal{D},\sigma)$ is infinitesimally joint rigid in $\mathbb{R}^{2+2}$. By Theorem \ref{thm_IJRrank}, $\rank(R_{cd}(\sigma))=4n-5$ and $\Null(R_{cd}(\sigma))=\Span{u_1,u_2,u_3,u_4,u_5}=\Span{[x_1^T,\allzero^T]^T,[x_2^T,\allzero^T]^T,[x_3^T,\allzero^T]^T,[\allzero^T,y_4^T]^T,[p^T,\varphi^T]^T}$ where $x_1=\allone_n\otimes[1,0]^T$, $x_2=y_4=\allone_n\otimes[0,1]^T$, $x_3=(I_n\otimes J^1_2)p$. Since elementary row operations preserve the null space, $\Null(R_{cd})={\Null(TR_{cd})}$.

Consider a nonzero vector $w\in\Null(R_d(p))$. Let $u=[w^T,\allzero^T]^T\in\mathbb{R}^{4n}$, then by Lemma \ref{lem_TRcd}, $u\in \Null(TR_{cd}(\sigma))$. Hence $u$ must be a linear combination of $u_1$, $u_2$, $u_3$, $u_4$ and $u_5$, i.e., for some scalar $a_1$, $a_2$, $a_3$, $a_4$ and $a_5$, not all zero, $u=a_1u_1+a_2u_2+a_3u_3+a_4u_4+a_5u_5$. Examining the last $2n$ rows of $u$ gives $\allzero=a_4y_4+a_5\varphi$. Vectors $y_4$ and $\varphi$ are linearly independent, hence $a_4=a_5=0$ and $a_1,a_2,a_3$ are not all zero. For the remaining $2n$ rows, 
\begin{align}
    w=a_1x_1+a_2x_2+a_3x_3.
    \label{eqn_nullRd}
\end{align}
Since (\ref{eqn_nullRd}) holds for any $w\in \Null(R_d)$ and $x_1,x_2,x_3$ are linearly independent, then $\rank(R_d)=2n-3$ and $\Null(R_d)=\Span{x_1,x_2,x_3}$, i.e., the position framework $(\mathcal{G},p)$ is infinitesimally distance rigid. 

Under Assumption 1, the number of edges in $\mathcal{D}$ must be even. Since $\rank(R_{cd}(\sigma))=4n-5$, the joint rigidity matrix $R_{cd}$ should have at least $(4n-4)$ rows. So the distance rigidity matrix $R_d(p)$ has at least $(2n-2)$ rows, i.e., there are at least $(2n-2)$ edges in $\mathcal{G}$. As only $(2n-3)$ edges are necessary for infinitesimal distance rigidity, then the position framework $(\mathcal{G},p)$ is infinitesimally distance rigid with at least one redundant edge, so (a)$\Rightarrow$(c).
\end{proof}

Theorem \ref{thm_joint2D_NS} proves the equivalence between clock rigidity and joint rigidity for $d=2$. So the topological method to establish clock rigidity in Corollary \ref{cor_ICR_Laman} is also applicable to joint rigidity. 

The necessary and sufficient condition in Theorem \ref{thm_joint2D_NS} can not be extended to $d=3$. This follows as that statement (b)$\Rightarrow$(c) does not hold for $d=3$ since distance rigidity in $\mathbb{R}^3$ with at least one redundant edge (whose minimum edge number is $3n-5$) is not a necessary condition for bearing rigidity in $\mathbb{R}^2$ with at least one redundant edge (whose minimum edge number is $2n-2$). It can be seen that for $n\geq4$, $2n-2<3n-5$. A cardinality argument can be used to show the contradiction. Further, statement (a)$\Rightarrow$(c) also does not hold for $d=3$. A simple counterexample is that for $n=4$, by Corollary \ref{cor_minimumedges}, the minimum edge number for a joint rigid graph is $12$, which is also the minimum edge number for distance rigidity in $\mathbb{R}^3$, which does not satisfy the redundant edge requirement in statement (c).

Many wireless sensor network applications require position estimation in $3$-dimensional space. So we establish the following theorem, which shows a sufficient condition for infinitesimal joint rigidity in $\mathbb{R}^{3+2}$.

\begin{thm}
Under Assumptions \ref{assump1} and \ref{assum2}, for any generic clock configuration $\varphi$, a position-clock framework $(\mathcal{D},\sigma)$ in $\mathbb{R}^{3+2}$ with $\sigma=[p^T,\varphi^T]^T$ is infinitesimally joint rigid if the position framework $(\mathcal{G},p)$ in $\mathbb{R}^3$ is infinitesimally distance rigid and $n\geq 4$, where $\mathcal{G}$ is the disoriented graph of $\mathcal{D}$.
\end{thm} 

\begin{proof}\label{thm_joint3D}
By Theorem \ref{thm_IDR+ICR->IJR}, we only need to show that $(\mathcal{G},\varphi)$ is infinitesimally clock rigid.

Since $(\mathcal{G},p)$ is infinitesimally distance rigid in $\mathbb{R}^3$, there must exist a subgraph $\mathcal{G}_s=(\mathcal{V},\mathcal{E}_s)$ with $|\mathcal{E}_s|=3|\mathcal{V}|-6$, such that every subset of $k\geq 3$ vertices spans at most $3k-6$ edges \cite{anderson2008rigid}. In other words, $\mathcal{G}_s$ can be formed from a triangle graph by Henneberg construction in $\mathbb{R}^3$, including vertex addition in $\mathbb{R}^3$ which adds a new vertex with three new incident edges, and edge splitting in $\mathbb{R}^3$ which removes an existing edge and adds a new vertex with four new incident edges. 

Now we can reconstruct the graph from a triangle graph by replacing the Henneberg operations in $\mathbb{R}^3$ by the corresponding Henneberg operations in $\mathbb{R}^2$, i.e., add one less edge in every operation (add two less edges if the edge to be removed does not exist). Then for $n\geq 4$, the resulting graph $\mathcal{G}'_s\subset \mathcal{G}_s\subseteq\mathcal{G}$ and $\mathcal{G}'_s$ is a Laman graph.

By Laman's theorem, with generic $\varphi$, a clock framework $(\mathcal{G},\varphi)$ is infinitesimally distance rigid if and only if a subgraph of $\mathcal{G}$ is a Laman graph \cite{tay1985generating}. So, $(\mathcal{G},\varphi)$ is infinitesimally distance rigid. Since for $n\geq 4$, $\mathcal{G}'_s\subset \mathcal{G}$, there must exist at least one redundant edge for distance rigidity. By Theorem \ref{thm_GBR<->ICR}, the clock framework $(\mathcal{G},\varphi)$ is infinitesimally clock rigid. 
\end{proof}

The analysis and results of joint rigidity share many similarities with distance, bearing and clock rigidity theory. Their parallels are summarized in Table \ref{tab_rigiditysummary}.


\section{Joint position and clock estimation} \label{sec_estimation}
Joint rigidity theory and the corresponding results show the graph property with which the nodes' position and clock in a TOA-based UWB sensor network can be determined up to some trivial variations, i.e., translation and rotation of position configuration $p$, translation of the clock offset configuration $\beta$ and a scaling of entire position-clock framework. In this section, we study the position and clock estimation of a UWB sensor network based on the clock rigidity theory and demonstrate through simulation.

\subsection{Clock estimation}
Let $\hat{\varphi}$ be an estimation of the true clock configuration $\varphi$. We consider the estimation error
\begin{align}
    e_c(\hat{\varphi},\varphi)=F_c(S(\hat{\varphi}),\hat{\varphi})-F_c(S(\varphi),\varphi),
\end{align}
where $F_c$ is the clock function defined in (\ref{clockfunc}). Since the clock configuration of the network is assumed to be constant, by (\ref{clockfunc_constr}), $F_c(S(\varphi),\varphi)=0$. We write the estimation error as $e_c(\hat{\varphi})=F_c(S(\hat{\varphi}),\hat{\varphi})$ for simplicity.

The objective of the clock estimation can be stated as the minimization of the following function
\begin{align}\label{obj_1}
    P_1(\hat{\varphi})=\frac{1}{2}\norm{e_c(\hat{\varphi})}^2=\frac{1}{2}\sum_{\{v_i,v_j\}\in\mathcal{E}}e_{c_{ij}}(\hat{\varphi})^2,
\end{align}
where $e_{c_{ij}}(\hat{\varphi})=\hat{\alpha}_j\overline{T}^j_{ij}+\hat{\beta}_j-\hat{\alpha}_i\overline{T}^i_{ij}-\hat{\beta_i}$. The minimization of (\ref{obj_1}) can be obtained by the gradient descent method
\begin{align}\label{GD_1}
    \dot{\hat{\varphi}}=-k_g\frac{\partial P_1(\hat{\varphi})}{\partial \hat{\varphi}}=-k_gR_c(\hat{\varphi})^Te_c(\hat{\varphi}),
\end{align}
where $k_g$ is a positive gain.

\begin{figure}[t]
\begin{center}
\subfigure[Clock estimation]{\includegraphics[width=2.5in]{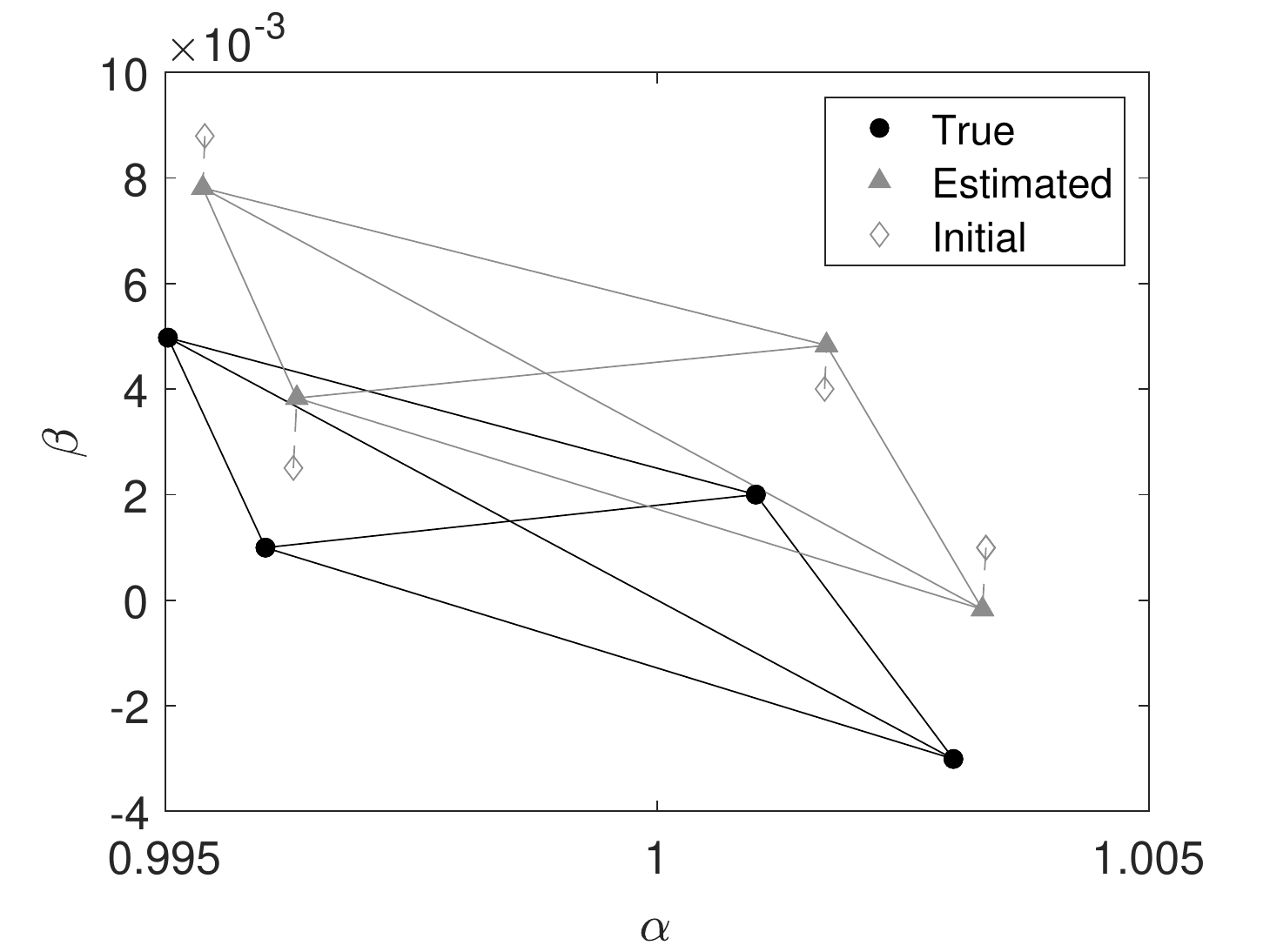}%
\label{Fig_clockest}}
\hfil
\subfigure[Estimation error $\norm{e_{c_{ij}}}$]{\includegraphics[width=2.5in]{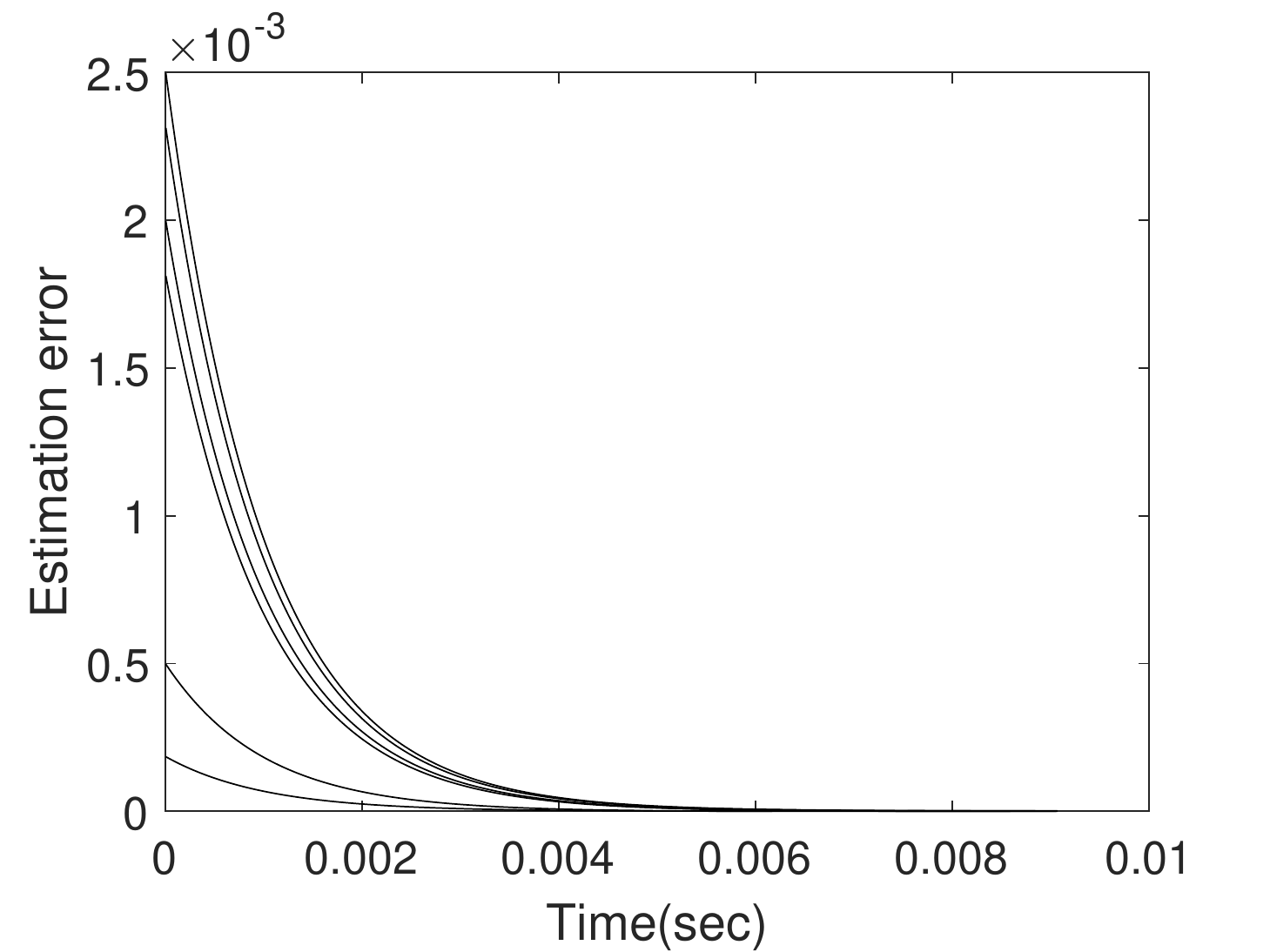}%
\label{Fig_clockest_error}}
\end{center}
\caption{Simulation results of clock estimation. Subplot (a) shows the initial clock configuration (diamond), the estimated clock configuration (triangle) and the true clock configuration (circle) of an infinitesimally clock rigid framework $(K_4,\varphi)$. Subplot (b) shows the behavior of the estimation error on each edge.}
\label{Fig_clockest_sim}
\end{figure}

If a clock framework $(\mathcal{G},\varphi)$ is infinitesimally clock rigid then for any sufficiently small neighborhood around the true clock configuration $\varphi$, the clock estimate $\hat{\varphi}$ converges to a set where $P_1(\hat{\varphi})=0$. This follows from LaSalle's invariance principle and the semidefiniteness of $R_c(\hat{\varphi})^TR_c(\hat{\varphi})$. So the clock estimate $\hat{\varphi}$ must be a trivial variation of the true configuration (a translation of clock offset configuration $\beta$ and a scaling of the entire clock framework), i.e., the estimated $\hat{\varphi}$ should reach a clock configuration such that
\begin{align}
    \hat{\varphi}=k_s\varphi+\allone_n\otimes[0,k_\beta]^T,
\end{align}
where $k_s$ is the scaling factor and $k_\beta$ is the translation factor of $\beta$. Simulation results are shown in Fig. \ref{Fig_clockest_sim}. The clock framework $(K_4,\varphi)$ is infinitesimally clock rigid as shown in Example \ref{exmp_K3andK4}. Given a random initial configuration in the sufficiently small neighborhood of true configuration, the estimated configuration is a trivial variation (translation of $\beta$ and scaling of the true configuration) and estimation errors converge to zero.

\begin{table*}[t]
  \centering
  \caption{Summary of different rigidity theories}
 \label{tab_rigiditysummary}
 \begin{tabular}{ |c||c|c|c|c| } 
\hline
  & Distance rigidity & Bearing rigidity & Clock rigidity & Joint rigidity \\ 
 \hline
  Framework & \multicolumn{2}{c|}{$(\mathcal{G},p)$, $p_i\in\mathbb{R}^d$} & $(\mathcal{G},\varphi)$, $\varphi_i\in\mathbb{R}^2$ & $(\mathcal{G},\sigma)$, $\sigma_i\in\mathbb{R}^{d+2}$ \\
  \hline
 Measurement & distance $d_{ij}$ & bearing $b_{ij}$ & timestamps  $T_{(i,j)}^i,T_{(i,j)}^j$ &  timestamps $T_{(i,j)}^i,T_{(i,j)}^j$\\
 \hline
 Invariant equality & $\norm{p_i-p_j}=d_{ij}$ & $\frac{p_i-p_j}{\norm{p_i-p_j}}=b_{ij}$  & Equation (\ref{clock_constr}) &  Equation (\ref{constr_tof})  \\
 \hline
 Rigidity matrix & $R_d$ & $R_b$ & $R_c$ & $R_{cd}$ \\
 \hline
  Trivial infinitesimal variations & translations, rotations & translations, scaling & translation of $\beta$, scaling & \makecell{translations of $p$, rotations of $p$,\\translation of $\beta$, scaling of $\sigma$}\\
  \hline
 Infinitesimal rigidity & \makecell{$\rank(R_d)=$\\$dn-d(d+1)/2$} & \makecell{$\rank(R_b)=$\\$dn-d-1$} &  \makecell{$\rank(R_c)=$\\$2n-2$}& \makecell{$\rank(R_{cd})=$\\$(d+2)n-d(d+1)/2-2$ }\\
 \hline
 Minimum rigid graph for $d=2$& \multicolumn{2}{c|}{Laman Graph}  & \multicolumn{2}{c|}{Laman Graph with one redundant edge}\\
\hline
 \end{tabular}
\end{table*}

\subsection{Joint position and clock estimation}
Joint position and clock estimation follows similarly to clock estimation above. Let $\hat{\sigma}=[\hat{p}^T,\hat{\varphi}^T]^T$ be an estimation of the true position-clock configuration $\sigma=[p^T,\varphi^T]^T$ and denote the estimation error as $e(\hat{\sigma})$ with elements $e_{ij}(\hat{\sigma})=\norm{\hat{p}_i-\hat{p}_j}^2-c^2(\hat{\alpha}_jT^j_{(i,j)}+\hat{\beta}_j-\hat{\alpha}_iT^i_{(i,j)}-\hat{\beta_i})^2$ for all $(v_i,v_j)\in\mathcal{E}_\mathcal{D}$. Following from the constraint in (\ref{constr_tof}), the estimation objective function is 
\begin{align}\label{obj_2}
  P_2(\hat{\sigma})=\frac{1}{4}\norm{e(\hat{\sigma})}^2=\frac{1}{4}\sum_{(v_i,v_j)\in\mathcal{E}_\mathcal{D}}e_{ij}(\hat{\sigma})^2.
\end{align}

\begin{figure*}[t]
\centerline{
\subfigure[Joint estimation: position]{\includegraphics[width=0.333\textwidth]{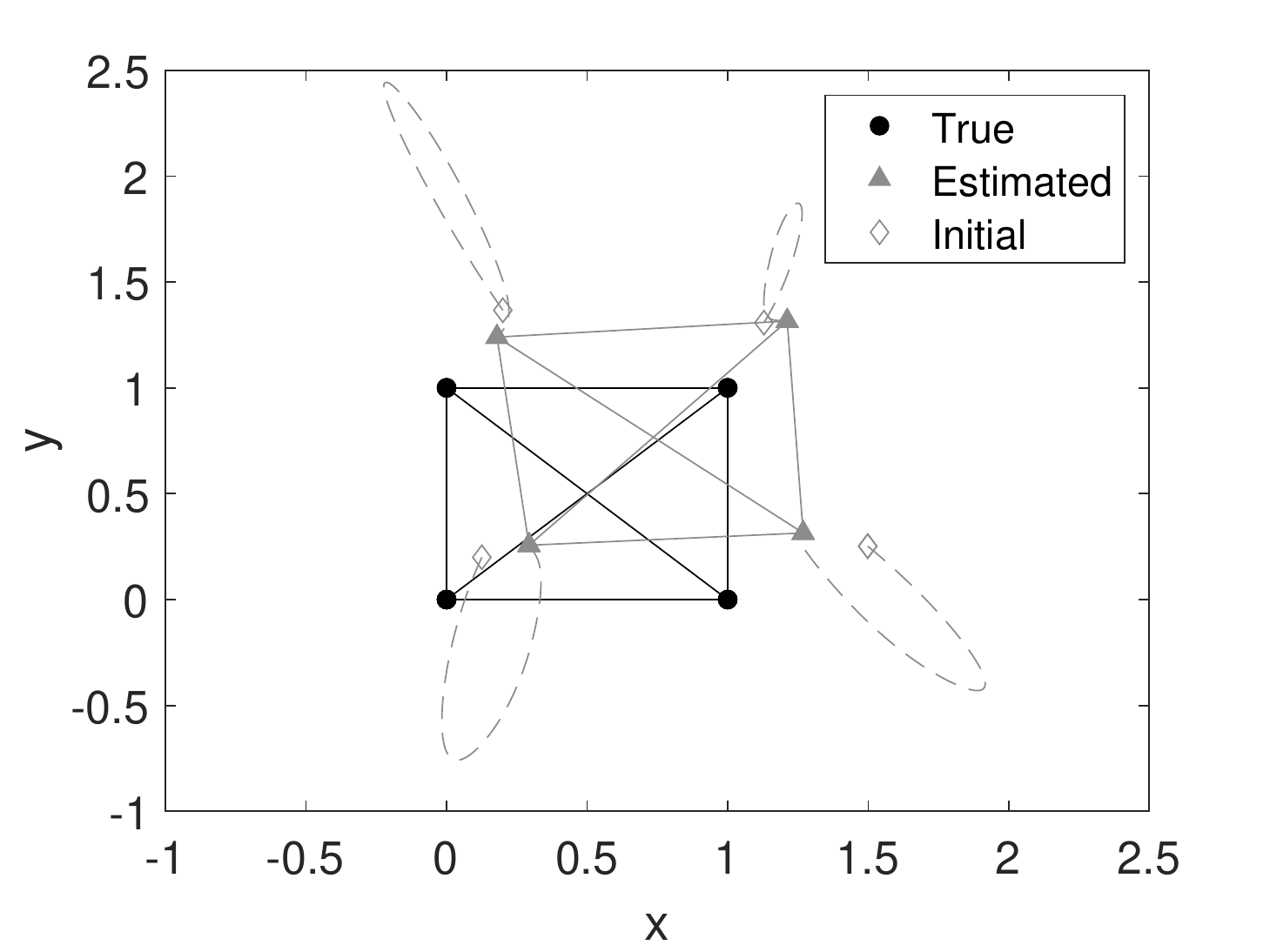}%
\label{Fig_joint_posest}}
\hfill
\subfigure[Joint estimation: clock]{\includegraphics[width=0.333\textwidth]{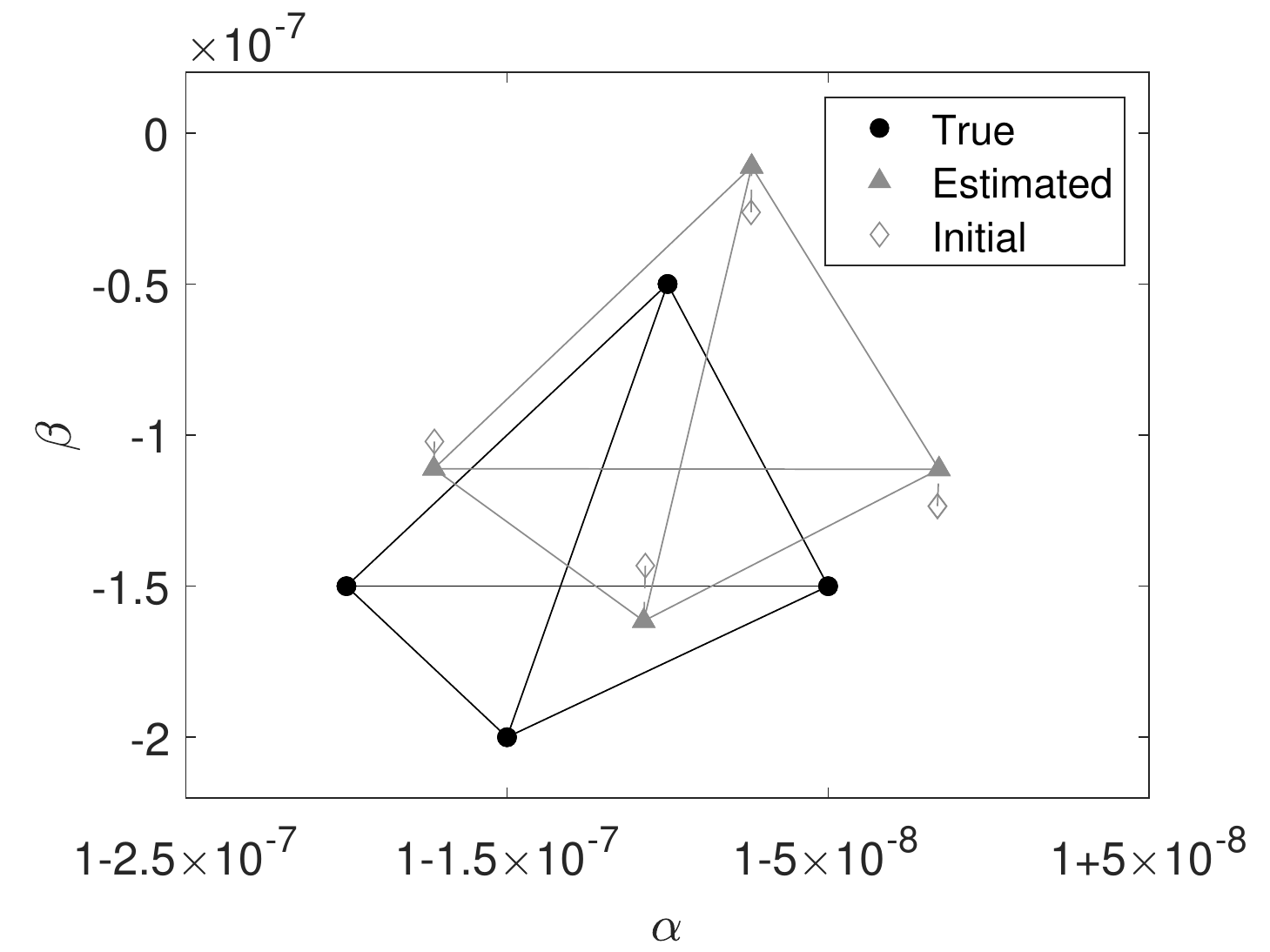}%
\label{Fig_joint_clockest}}
\hfill
\subfigure[Joint estimation error $\norm{e_{ij}}$ ]{\includegraphics[width=0.333\textwidth]{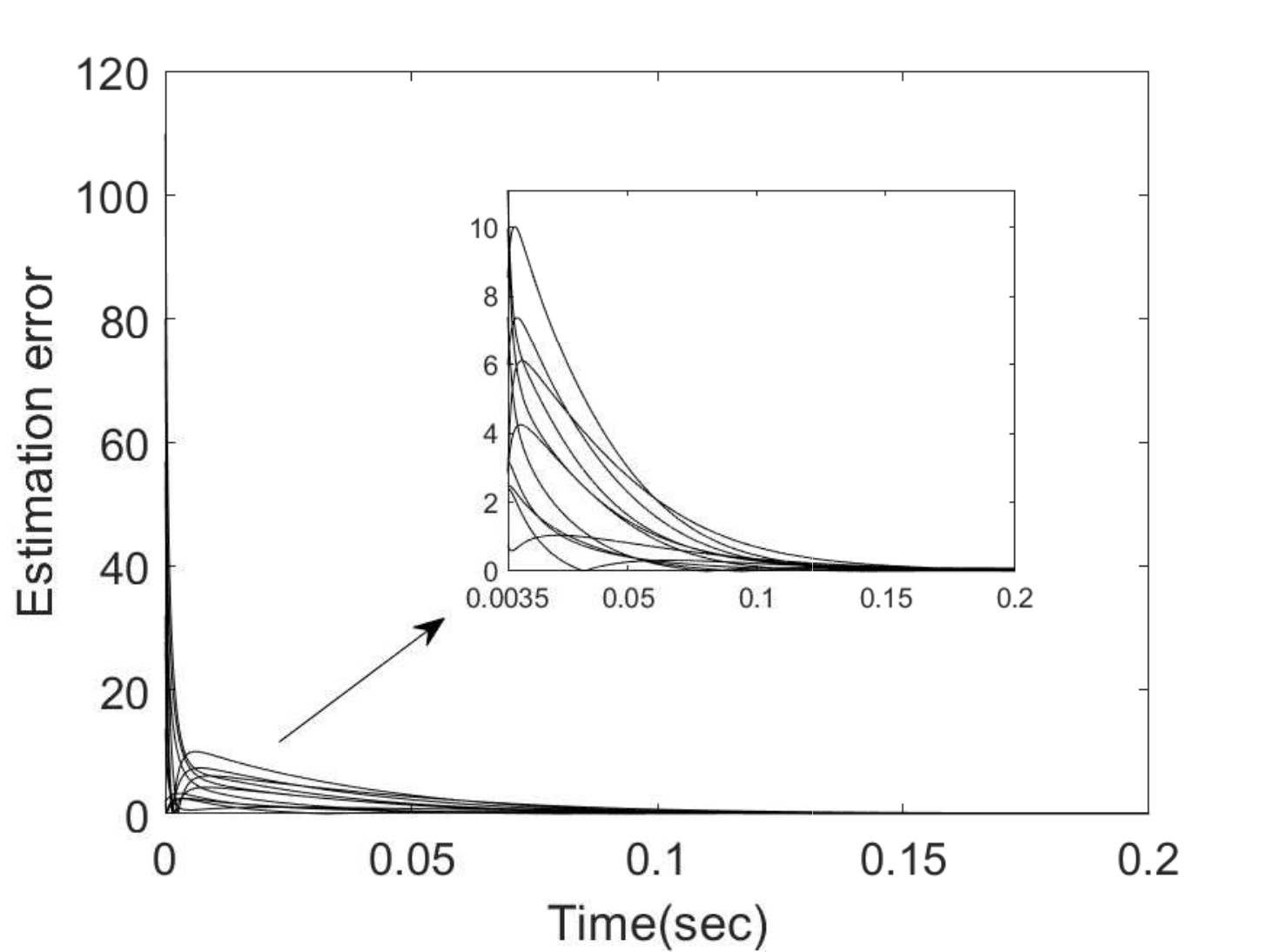}%
\label{Fig_joint_err}}
}
\caption{Simulation results of joint position and clock estimation. Subplot (a) and (b) show the initial configuration (diamond), the estimated configuration (triangle) and the true configuration (circle) of an infinitesimally position-clock rigid framework $(K_4,\sigma)$ in $\mathbb{R}^{2+2}$. Subplot (c) shows the behavior of the joint estimation error on each edge.}
\label{Fig_Jointest_sim}
\end{figure*}

The objective function is equal to zero if and only if $e_{ij}(\hat{\sigma})=0$ for all $(v_i,v_j)\in\mathcal{E}_\mathcal{D}$. The minimization of (\ref{obj_2}) can be obtained by the gradient descent method
\begin{align}\label{GD_2}
    \dot{\hat{\sigma}}=-K_g\frac{\partial P_2(\hat{\sigma}),\hat{\sigma})}{\partial \hat{\sigma}}=-K_gR_{cd}(\hat{\sigma})^Te(\hat{\sigma}),
\end{align}
where $K_g$ is a diagonal matrix whose diagonal entries are positive gains for the position variables and clock variables, respectively. The position configuration and clock configuration are usually at different measurement scales. In order to determine a suitable step size for the gradient descent method, we use a diagonal matrix $K_g$ to choose appropriate gains for position and clock terms.

If a clock framework $(\mathcal{D},\sigma)$ is infinitesimally joint rigid then in any sufficiently small neighborhood of the true position-clock configuration $\sigma$, the position-clock estimate $\hat{\sigma}$ converges to a set where $P_2(\hat{\sigma})=0$ following from LaSalle's invariance principle and the semidefiniteness of $R_{cd}(\hat{\sigma})^TR_{cd}(\hat{\sigma})$. The position-clock estimate $\hat{\varphi}$ must be a trivial variation of the true position-clock configuration (translation and rotation of position configuration $p$, translation of the clock offset configuration $\beta$ and a scaling of entire position-clock framework). Therefore, in $2$-dimensional space for example, the estimated $\hat{\sigma}=[\hat{p}^T,\hat{\varphi}^T]^T$ should reach a position-clock configuration such that
\begin{equation}\label{joint_est}
\begin{split}
\hat{p}&=k_sp+k_r(I_n\otimes J_2^1)p+\allone_n\otimes k_t\\
    \hat{\varphi}&=k_s\varphi+\allone_n\otimes[0,k_\beta]^T,
\end{split}
\end{equation}
where $k_s$ is the scaling factor, $k_r$ is the rotation factor, $k_t\in\mathbb{R}^2$ is the translation factor of $p$ and $k_\beta$ is the translation factor of $\beta$. 

Under Assumption 1, following from Theorem \ref{thm_joint2D_NS}, the estimated $\hat{\sigma}$ in $2$-dimensional space follows (\ref{joint_est}) if and only if the disoriented graph of $\mathcal{D}$ is generically distance rigid with at least one redundant edge. Simulation results are shown in Fig. \ref{Fig_Jointest_sim}. The position-clock framework $(\mathcal{D},\sigma)$ is infinitesimally joint rigid in $\mathbb{R}^{2+2}$, so given a random initial configuration in the sufficiently small neighborhood of true configuration, the estimated configuration is a trivial variation of the true configuration and estimation errors converge to zero. As can be seen in Fig. \ref{Fig_Jointest_sim}, the variation is a combination of translation and rotation of $p$, translation of $\beta$ and scaling of the entire framework, where the scaling is obvious for clock in Fig.\ref{Fig_Jointest_sim}(b) whereas not obvious for position in Fig.\ref{Fig_Jointest_sim}(a) due to the figure scale.

\section{Conclusion}
In this paper we proposed a clock rigidity theory for TOA-based UWB sensor network, showing that a clock framework with certain graph properties can be uniquely determined up to some trivial variations (a shift on clock offset and a skew on all clock parameters) given the TOA timestamp measurements. We also showed that a clock framework is infinitesimally clock rigid if and only if its underlying graph is generically bearing rigid in $\mathbb{R}^2$ with at least one redundant edge, providing a topological method for the analysis of clock rigidity.

Building on the proposed clock rigidity theory, we studied the joint position and clock estimation problem. We similarly proved that a position-clock framework with certain graph properties can be uniquely determined up to some trivial variations corresponding to both position and clock, i.e., a translation and a rotation of position, a shift of clock offset and a scaling of both position and clock. 

Clock estimation and joint position-clock estimation method has been proposed and validated in the simulations.

The estimation considered in this paper is anchor-free. To uniquely determine a position-clock framework without any trivial variation, we need to have sufficient knowledge of the absolute position and clock information of some sensors in the network. In the future work, we will formulate the joint estimation problem in the presence of anchors and investigate how to select necessary anchors based on joint rigidity theory. 

The gradient-descent method used in Section \ref{sec_estimation} forms a distributed estimator over the sensor network for both clock estimation and joint position-clock estimation, but it also leads to slow convergence. Faster distributed estimation methods for network localization and time synchronization is another meaningful direction to be studied in the future.

\bibliographystyle{IEEEtran}
\bibliography{IEEEabrv,Bibliography}

\end{document}